\documentclass{article}
\usepackage[colorlinks,urlcolor=blue,linkcolor=blue,citecolor=blue]{hyperref}
\usepackage{cite}
\usepackage{amsmath,amsfonts,amssymb,amsthm,mathtools,hyperref,mathrsfs,bbm}
\usepackage{graphicx}
\usepackage{algorithm,algorithmic}
\usepackage{graphicx,float,subcaption,color,import} 
\usepackage[top=0.75in, bottom=0.79in, left=0.75in, right=0.75in]{geometry}
\newtheorem{assumption}{Assumption}
\newtheorem{lemma}{Lemma}
\newtheorem{theorem}{Theorem}
\newtheorem{remark}{Remark}

\begin{document}

\title{\textbf{Formal Entropy-Regularized Control of Stochastic Systems}}

\author{M. J. T. C. van Zutphen, G. Delimpaltadakis, D. Antunes}

\date{}

\maketitle

\textbf{\textit{abstract:}} Analyzing and controlling system entropy is a powerful tool for regulating predictability of control systems. Applications benefiting from such approaches range from reinforcement learning and data security to human-robot collaboration. In continuous-state stochastic systems, accurate entropy analysis and control remains a challenge. In recent years, finite-state abstractions of continuous systems have enabled control synthesis with formal performance guarantees on objectives such as stage costs. However, these results do not extend to entropy-based performance measures. We solve this problem by first obtaining bounds on the entropy of system \textit{discretizations} using traditional formal-abstractions results, and then obtaining an additional bound on the difference between the entropy of a continuous distribution and that of its discretization. The resulting theory enables formal entropy-aware controller synthesis that trades predictability against control performance while preserving formal guarantees for the original continuous system. More specifically, we focus on minimizing the linear combination of the KL divergence of the system trajectory distribution to uniform --- our system entropy metric --- and a generic cumulative cost. We note that the bound we derive on the difference between the KL divergence to uniform of a given continuous distribution and its discretization can also be relevant in more general information-theoretic contexts. A set of case studies illustrates the effectiveness of the method.






\section{Introduction}
\label{sec:introduction}

In stochastic systems, the concept of \emph{entropy} provides a quantitative measure of predictability, capturing how random or predictable the potential behavior of a system is. As a result, entropy optimization --– maximizing or minimizing system entropy --– has emerged as a powerful concept across several fields~\cite{Kuiper2020UnpredictableMotion,Wijlens2024RoadToComfort,Schwarting2019SocialBehavior,Ornia2025PredictableRL,vanZutphen2024PredictableIMDPs,biondi2015quantifying,guo2023balancing,Guo2023EMPatrollerEM,Liu2020Ready,Liu2024Enhancing,haarnoja17reinforcement,haarnoja18bSA,ahmed19aUnderstandingTI}. Increased predictability can improve passenger comfort~\cite{Kuiper2020UnpredictableMotion} and disturbance robustness~\cite{Wijlens2024RoadToComfort} in self-driving cars, and cooperation~\cite{Schwarting2019SocialBehavior,Ornia2025PredictableRL,vanZutphen2024PredictableIMDPs} in autonomous mobile agents more generally. In contrast, a high level of unpredictability is crucial for confidentiality in data security applications~\cite{biondi2015quantifying}, surveilance tasks~\cite{guo2023balancing,Guo2023EMPatrollerEM}, preventing exploitation of predictable self-driving car behavior~\cite{Liu2020Ready,Liu2024Enhancing}, and the promotion of exploration and improved learning robustness in reinforcement learning (RL)~\cite{haarnoja17reinforcement,haarnoja18bSA,ahmed19aUnderstandingTI,Ornia2025PredictableRL,ashlag2025state}.

Quantification of predictability in finite-state stochastic systems generally relies on the concept of \emph{information entropy} as introduced by Shannon~\cite{Shannon1948,Cover1991ElementsOI}. While recent RL literature often considers rewarding (or penalizing) the level of randomness in its input selection, i.e., \emph{policy entropy}~\cite{haarnoja17reinforcement,haarnoja18bSA,ahmed19aUnderstandingTI}, our work focuses on quantification and optimization of the level of randomness in system behavior, i.e., \emph{trajectory entropy}. Efficient analysis of trajectory entropy in Markov systems is generally non-trivial, and has thus been the subject of ongoing research, which we briefly review below.

Early work on trajectory entropy optimization considers the problem of finding a maximal-entropy Markov chain (MC) consistent with a given interval Markov chain (IMC) specification; a generalization of an MC in which transition probabilities are constrained to lie within prescribed uncertainty intervals~\cite{Biondi2014MaximizingEO}. Subsequent work~\cite{Chen2014OnTC} investigates the complexity of computing the (maximum) entropy and entropy rate for MCs, IMCs and Markov decision processes (MDPs). 
Follow up work later extend this MDP method~\cite{Chen2014OnTC} to include temporal logic constraints~\cite{Savas2020EntropyMDP_TL,Chen2022EntropyRM_LTL}, developing algorithms that can constructively obtain associated optimal policies, and applying them to a practical surveillance task~\cite{Chen2023EntropyRM}. A Bellman equation balancing cost with trajectory entropy, and an associated RL framework that promotes unbiased exploration by maximizing trajectory entropy under cost constraints are obtained in~\cite{Srivastava2021ParameterizedMDPs}. Policies that minimize a formal upper bound on an entropy-regularized objective are obtained for uncertain systems formulated as interval MDPs (IMDPs) in~\cite{vanZutphen2024PredictableIMDPs}.
Efficient RL exploration through maximum-entropy policies is explored~\cite{tiapkin23a}, 
again showing that the trajectory-entropy maximizing policy can be found by solving an entropy-regularized Bellman equation. Trajectory entropy penalization in RL is developed in~\cite{Ornia2025PredictableRL}, yielding agents that humans and other agents can more easily anticipate. A very similar approach minimizes the \emph{action-trajectory} entropy instead~\cite{you2025trajectory}.

These existing methods are all fundamentally tied to the finite-state setting, while many systems of interest naturally evolve over continuous domains. Abstraction methods have been developed to enable formal analysis and control of continuous-state stochastic systems with complex specifications through carefully constructed finite-state representations~\cite{Soudjani2015AggregationAC,HaesaertSoudjaniAbate2017,Hahn2019IntervalMD,lavaei2022automated,dutreix2022abstraction,BadingsRomaoAbateJAIR23,Delimpaltadakis2023IntervalMD,Mathiesen2024IntervalMDPjlAV,Suilen2024RobustRMDP,Banse2025MemoryDA,Delimpaltadakis2024FormalAO}. By discretizing the state and action spaces of the original system, a finite-state model --- the abstraction --- is constructed in a \emph{formal way}; that is, formal guarantees accompanying verification and control design over the abstraction carry over to the original system (e.g., if the abstraction is ``safe'', then the original system is ``safe''). Interval Markov decision processes (IMDPs), which generalize MDPs by allowing transition probabilities to vary within intervals, have emerged as a standard abstraction formalism for (uncertain) continuous Markov systems~\cite{vanZutphen2024PredictableIMDPs,Soudjani2015AggregationAC,HaesaertSoudjaniAbate2017,Hahn2019IntervalMD,lavaei2022automated,dutreix2022abstraction,BadingsRomaoAbateJAIR23,Delimpaltadakis2023IntervalMD,Mathiesen2024IntervalMDPjlAV,Givan2000BoundedparameterMD}. These methods enable sound and computationally tractable guarantees for verification and synthesis tasks through finite models in general settings.

However, the unique challenge of entropy minimization or maximization does not naturally integrate into these existing abstraction frameworks the way expected cumulative stage costs~\cite{Givan2000BoundedparameterMD,HaesaertSoudjaniAbate2017,lavaei2022automated,BadingsRomaoAbateJAIR23,Delimpaltadakis2023IntervalMD,Mathiesen2024IntervalMDPjlAV}, or LTL specifications~\cite{Soudjani2015AggregationAC,Lahijanian2015FormalVA,LaurentiLahijanianAbateCardelliKwiatkowska2021,lavaei2022automated,BadingsRomaoAbateJAIR23,BadingsJansenRomaoAbate2023Correct} do. In this work, we show that existing abstraction-based approaches applied to system entropy fail to preserve or bound entropy properties of the underlying system. We then propose a solution in the form of an adapted framework that addresses these shortcomings. 

The main purpose of this paper is thus to develop a formal theory on \emph{abstraction-based entropy guarantees}, i.e., to enable the formal analysis and control of continuous-state Markov system entropy through finite-state abstractions. We quantify entropy in our framework through the Kullback Leibler (KL) divergence~\cite{Cover1991ElementsOI} of the trajectory distribution to the uniform distribution (hereafter referred to as the \emph{KL divergence to uniform}), a well-behaved and lossless proxy for system entropy.

We first obtain a number of results that enable the computation of formal lower-bounds on the KL divergence to uniform of a continuous MC trajectory through IMC-abstractions. We then extend these results to additionally construct, under appropriate assumptions, formal upper-bounds on the same KL divergence to uniform. In the process, we derive bounds on the difference between the KL divergence to uniform of a given continuous distribution, and the KL divergence to uniform of its discretizations, which can have their own standalone significance in the broader information-theoretic context. 

Afterwards, we extend the method to abstraction-based \emph{policy synthesis} for continuous-state MDPs. Having obtained two distinct approaches toward the construction of formal upper bounds above, we present an algorithm that finds policies minimizing each. We conclude by demonstrating the effect of discretization resolution on bound tightness in a numerical example. We further consider a continuous-state control problem where we regularize the cumulative cost optimization by an entropy penalty. The system behavior resulting from the policies we generate realizes a reduced level of entropy in the resulting system trajectories.

The organization of the paper is as follows. In Section~\ref{sec:markov_system}, we introduce the prerequisite preliminary concepts. We introduce the continuous-state Markov chain and present two approaches to its discretization: the mean-value trajectory distribution discretization and the interval Markov chain (IMC) abstraction. We then briefly discuss the quantification of system entropy in the context of compact continuous-state Markov systems. In Section~\ref{sec:bounds_sys_ent_abstractions}, we develop our main results on system-entropy (as expressed through the KL divergence to uniform of the trajectory distribution) guarantees on continuous-state Markov systems through finite-state abstractions. In Section~\ref{sec:generalizing_to_control}, we extend these results to Markov decision processes (MDPs) and present a bound-minimizing policy synthesis algorithm. In Section~\ref{sec:examples}, we present the results of two numerical experiments that demonstrate the efficacy of the system entropy abstraction approaches developed above. We finish by presenting the conclusions in Section~\ref{sec:conc_and_discussion}.

\section{Markov chains and entropy}\label{sec:markov_system}

Let $\lambda$ denote the Lebesgue measure on $\mathbb{R}^n$. For a measurable set $A \subseteq \mathbb{R}^n$, we define the space of probability density functions over $A$ as
\[
\mathcal{P}(A) \textstyle :=\{ q : A \to \mathbb{R}_{\ge 0} \mid q \text{ measurable}, \int_{A} q(x)~\mathrm{d}x = 1 \},
\]
i.e., all functions $q \in \mathcal{P}(A)$ admit probability measures $\mu(B) = \int_B p(x)~\mathrm{d}x$ for all measurable $B \subseteq A$, where $\mu$ is \emph{absolutely continuous w.r.t.\ }$\lambda$.

Let $\mathbb{P}^{n}$ further denote the standard $(n-1)$-dimensional probability simplex as
\[
\mathbb{P}^{n} := \{ p \in \mathbb{R}^{n}_{\ge 0} \;|\; \textstyle\sum_{i=1}^{n} p_i = 1 \},
\]
i.e., the set of all probability distributions over $n$ elements.

For the sake of simplicity, we make use of the convention that $0\log0:=\lim_{x\downarrow 0}x\log x=0$.

\subsection{Markov chains}

Let a discrete-time stochastic dynamical system $\mathcal{M}$ evolve on an $n_{x}$-dimensional hyperrectangular state space $\mathcal{X}:=[a_1,b_1]\times[a_2,b_2]\times\cdots\times[a_{n_{x}},b_{n_{x}}]\subset\mathbb{R}^{n_{x}}$ at time steps $k\in\{0,1,\dots,K\}$, for $K\in\mathbb{N}$ --- this is without loss of generality; the results apply to any compact state space. Let its Markovian dynamics over horizon $K$ be described by the state transition probability density (transition kernel)
\[
\operatorname{Prob}
 (x_{k+1}\in A \ | \ x_k)= \int_{A}q(x_k,x_{k+1})~\mathrm{d}x_{k+1},
\]
where $A\subseteq\mathcal{X}$, $x_k\in\mathcal{X}$ is the state of the system at timesteps $k\in\mathbb{N}\cup\{0\}$, combined with an initial state distribution density $q_0$, described as
\[
\operatorname{Prob}(x_0\in A)=\int_{A}q_0(x_0)~\mathrm{d}x_0.
\]
The introduction of \emph{actions} is postponed to Section~\ref{sec:generalizing_to_control}, in the context of controller design. Note that we have implicitly assumed absolute continuity of the probability measures underlying $q_0$ and $q$ w.r.t.\ the Lebesgue measure. 

Denoting the space of all possible trajectories by $\mathcal{S}:=\mathcal{X}^{K+1}$, which is a hyperrectangular compact subset of \mbox{$n$-dimensional} real space, with $n=(K+1)n_{x}$, system $\mathcal{M}$ admits a trajectory-level description through distribution $T$ over trajectories $s\in\mathcal{S}$, as
\begin{equation}
T(s)=q_0(x_0)\prod_{k=1}^{K}q(x_{k-1},x_k),\label{eq:T}
\end{equation}
where $s=x_0,\dots,x_{K}$, is a trajectory in $\mathcal{S}$.

\subsubsection{Trajectory distribution discretization}

Let $(\mathcal{X}_{i})_{i\in X}$ be a hyperrectangular grid-based partition of $\mathcal{X}$. As a consequence, every element of the state space partition $\mathcal{X}_i$, for $i\in X:=\{1,2,\dots,|X|\}$, but also every element of the induced partition of the trajectories $\mathcal{S}_t$ for $t\in S:=\{1,2,\dots,|S|\}$, is hyperrectangular, with $|S|=|X|^{K+1}$, see Fig.\ \ref{fig:illustration}. Let discrete probability distribution $p=\begin{bmatrix}p_1 & p_2 &\cdots &p_{|S|}\end{bmatrix}^\top$ denote the discretization of $T$ associated with partition $(\mathcal{S}_t)_{t\in S}$, as
\begin{equation}
    p_{t} := \int_{\mathcal{S}_t}T(s)~\mathrm{d}s,\label{eq:discretization_of_T}
\end{equation}
for all $t\in S$.

\def\svgwidth{0.495\textwidth} 
\begin{figure}[tb]
    \centering
    \input{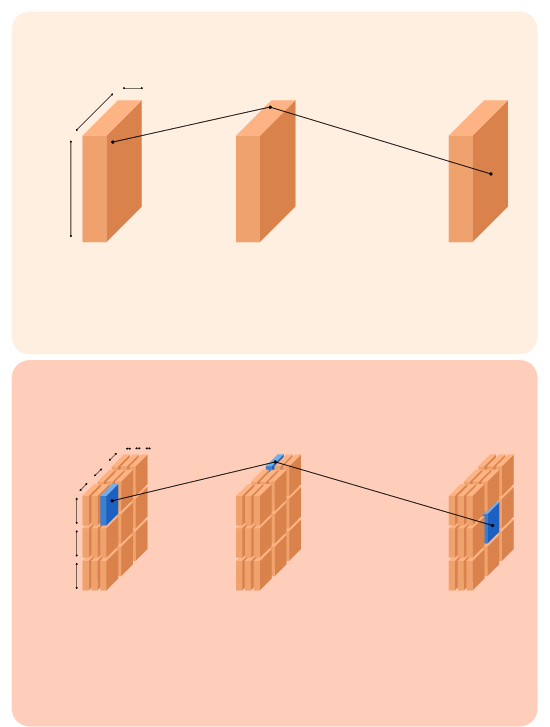}
    \caption{Discretization $(\mathcal{S}_t)_{t\in S}$ of the trajectory space $\mathcal{S}$, as induced by a hyperrectangular discretization $(\mathcal{X}_{i})_{i\in X}$ of the state space $\mathcal{X}$.}
    \label{fig:illustration}
\end{figure}

\subsubsection{Interval Markov chain abstraction}

Let the state space partition $(\mathcal{X}_i)_{i\in X}$ further induce a finite \emph{abstraction} of the original continuous space Markov process $\mathcal{M}$ as an interval Markov chain (IMC)~\cite{Biondi2014MaximizingEO} on state space $X$. Its initial finite-state distribution is then given by 
\begin{subequations}
\label{eq:abstract_IMC_definition}
\begin{equation}
\pi_{i}=\int_{\mathcal{X}_{i}}q_0(x_0)~\mathrm{d}x_0,
\end{equation}
for $i\in X$, while its (time-varying) Markov kernel is constrained to the interval
\begin{equation}
P_{k}(i)\in\mathbb{P}_{i}:=\{p^{\prime}\in\mathbb{P}^{|X|}\mid \underline{P}_{ij}\le p^{\prime}_{j}\le \overline{P}_{ij}\},\label{eq:mathbbPi}
\end{equation}
where
\begin{equation}
\overline{P}_{ij} = \max_{x\in \mathcal{X}_i}\int_{\mathcal{X}_j}q(x,x^{\prime})~\mathrm{d}x^{\prime}, \ \ \ \underline{P}_{ij} = \min_{x\in \mathcal{X}_i}\int_{\mathcal{X}_j}q(x,x^{\prime})~\mathrm{d}x^{\prime},
\end{equation}
\end{subequations}
for all $i,j\in X$, and all $k\in\{0,1,\dots,K\}$. Let $\mathbb{T}$ denote the set of all trajectory laws $p^{\prime}\in\mathbb{P}^{|S|}$, over $S=X^{K+1}$, admissible in IMC~\eqref{eq:abstract_IMC_definition}, i.e., satisfying initial distribution $\pi$, and admitting a time-varying Markov kernel sequence $P_{k}(i)\in\mathbb{P}_{i}$, for $i\in X$, $k\in\{0,1,\dots,K\}$.

In summary, we consider three distinct objects. The principal one being the continuous-state discrete-time MC $\mathcal{M}$, specified by either its state-space-level Markov kernel, or its trajectory-space-level probability distribution representation, as 
\begin{equation}
    \mathcal{M}=(\mathcal{X},q_0,q,K)\equiv(\mathcal{S},T). \label{eq:underlying_system_calM}
\end{equation}
Besides this \emph{``underlying''} system, we consider its trajectory-level discretization $p$~\eqref{eq:discretization_of_T} on discrete trajectory-space $S=X^{K+1}$, and the IMC abstraction defined by the set of time-varying MCs in $\mathbb{T}$~\eqref{eq:abstract_IMC_definition}.

\begin{assumption}\label{assumptions}
Our work relies on the following assumptions being satisfied:
\begin{itemize}
    \item $\mathcal{X}_i$ is hyperrectangular for every $i\in X$,
    \item $q$ and $q_0$ are continuous and differentiable almost everywhere,
    \item $\nabla q$ and $\nabla q_0$ are bounded as $\|\nabla q\|_{\infty}\le L_{\nabla q}$, $\|\nabla q_0\|_{\infty}\le L_{\nabla q}$ for some $L_{\nabla q}\in\mathbb{R}$ almost everywhere.
\end{itemize}
\end{assumption}

As a consequence of these assumptions, $T$ is bounded, continuous, and differentiable almost everywhere. $\nabla T$ is further bounded as $\|\nabla T\|_{\infty}\le L$, for some $L\in\mathbb{R}$ almost everywhere. All continuous distributions ($q$, $q_0$, $T$) can be written as densities.

\subsection{Entropy quantification} \label{sec:entropy_in_unc_spaces} 

System entropy is quantified by the \emph{KL divergence to uniform} of the distribution over trajectories of the system. For a distribution $T\in\mathcal{P}(\mathcal{S})$, this KL divergence to uniform is defined as
\begin{equation}
\operatorname{KL}(T\|U)=\int_{\mathcal{S}}T(s)\log \frac{T(s)}{1/\lambda(\mathcal{S})}~\mathrm{d}s,
\label{eq:T_to_uniform}
\end{equation}
where $U(s)=1/\lambda(\mathcal{S})$ is the uniform distribution over $\mathcal{S}$, and $\operatorname{KL}(\cdot \| U)$ is the KL divergence~\cite{Cover1991ElementsOI} of a distribution w.r.t.\ this uniform distribution. 

\begin{remark}
    The following relation is useful to understand the connection between the KL divergence to uniform of a distribution $T\in\mathcal{P}(\mathcal{S})$, and its differential entropy $h(T)$ 
    \[
    \begin{split}
        &\operatorname{KL}(T\|U) = \log\lambda(\mathcal{S}) +\int_{\mathcal{S}}T(s)\log T(s)~\mathrm{d}s \\
        &\hspace{4.5cm}=\max_{T^{\prime}\in\mathcal{P}(\mathcal{S})}~\{  h(T^{\prime})  \} - h(T),
    \end{split}
    \]
    as differential entropy $h(T) = -\int_{\mathcal{S}}T(s)\log T(s)~\mathrm{d}s$~\cite{Cover1991ElementsOI}.
\end{remark}

In discretized spaces, we quantify predictability through the \emph{discrete} KL divergence to uniform, as
\begin{equation}
\operatorname{KL}_{D}(p\|p^{\mathrm{u}})=\sum_{t\in S}p_t\log\frac{p_t}{p^{\mathrm{u}}_t}=\sum_{t\in S}p_t\log\frac{p_t\lambda(\mathcal{S})}{\lambda(\mathcal{S}_t)},
\label{eq:p_to_uniform}
\end{equation}
where $p$ and $p^{\mathrm{u}}$ are the discretizations of trajectory distribution $T$, and the uniform distribution $U$, according to~\eqref{eq:discretization_of_T}, respectively. 

Although the KL divergence to uniform and differential entropy are strongly related, their properties differ in a number of important ways. Most importantly, the Shannon entropy~\cite{Cover1991ElementsOI} of discrete distribution $p$~\eqref{eq:discretization_of_T} does not converge to the differential entropy of $T$ under increasingly fine discretizations, while we see that for shrinking $\overline{\lambda}(\mathcal{S}_t):=\max_{t\in S}\lambda(\mathcal{S}_t)$, we do obtain convergence in the KL divergence to uniform, as
\[
\begin{split}
    \hspace{2mm}\lim_{\mathclap{\overline{\lambda}(\mathcal{S}_t)\to 0}}\operatorname{KL}_{D}(p\|p^{\mathrm{u}}) &=\lim_{\overline{\lambda}(\mathcal{S}_t)\to 0} \sum_{t\in S}p_t\log\frac{p_t \lambda(\mathcal{S})}{\lambda(\mathcal{S}_t)} \\
    &=\lim_{\overline{\lambda}(\mathcal{S}_t)\to 0} \sum_{t\in S}\int_{\mathcal{S}_t}T(s)~\mathrm{d}s\log T(s_t) \lambda(\mathcal{S}) \\
    &=\lim_{\overline{\lambda}(\mathcal{S}_t)\to 0} \sum_{t\in S}\int_{\mathcal{S}_t}T(s)\log T(s_t) \lambda(\mathcal{S})~\mathrm{d}s \\
    &= \int_{\mathcal{S}}T(s)\log T(s) \lambda(\mathcal{S})~\mathrm{d}s \\
    &= \operatorname{KL}(T\|U),
\end{split}
\]
where we have used that $\lim_{\overline{\lambda}(\mathcal{S}_t)\to 0}\int_{\mathcal{S}_t}T(s)~\mathrm{d}s\to T(s_t)\lambda (\mathcal{S}_t)$ for any $s_t\in \mathcal{S}_t$, and all $t\in S$.

\section{Formally Bounding System Entropy} \label{sec:bounds_sys_ent_abstractions} 

In this section, we develop the theory required to obtain formal upper- and lower bounds on the KL divergence to uniform of the trajectories of original system $\mathcal{M}$, from an IMC abstraction $\mathbb{T}$.

In Section~\ref{sec:upper_bound}, we first obtain an approach that yields a lower bound on the KL divergence to uniform specifically, see Fig.~\ref{fig:flow_chart_bounds_LB}. Valid upper bounds require additional results, leading to two unique upper-bounding approaches. A ``global'' correction of standard IMC methods, which can be added a-posteriori to classic IMC algorithms to correct the bound they compute, and a more integrated ``local'' approach that augments the classic IMC algorithm itself, see Fig.~\ref{fig:flow_chart_bounds_UB}. Section~\ref{sec:practical_tools} then discusses techniques on how these results can be applied to actually compute bounds in practice.

\subsection{Bound construction}\label{sec:upper_bound}

From the abstraction completeness results for IMCs established in Thm.~3 and Prop.~3 of~\cite{Meng2022RobustlyCF}, it follows that for any hyperrectangular partition $(\mathcal{X}_{i})_{i\in X}$ of state space $\mathcal{X}$, the discretized trajectory distribution $p$ induced by~\eqref{eq:discretization_of_T}, is an element of the admissible IMC trajectory set $\mathbb{T}$, i.e., 
\begin{subequations}\label{eq:abstraction_completeness}
\begin{equation}
    p\in\mathbb{T}.
\end{equation}
Moreover, the abstraction is approximately complete in the sense that, as the partition is refined,
\begin{equation}
    \lim_{\overline{\lambda}(\mathcal{X}_{i})\to0}\mathbb{T}=\{p\}, 
\end{equation}
\end{subequations}
in the discretization resolution limit $\overline{\lambda}(\mathcal{X}_{i})\to 0$, where convergence is taken in the weak sense, i.e., all finite-time joint distributions converge (equivalently, all finite-dimensional marginals converge componentwise). This is fundamentally what enables existing IMC abstraction methods to obtain guarantees on discrete trajectory behavior $p$, and as a consequence $T$, while having access only to interval Markov chain $\mathbb{T}$. 

Letting
\[
\overline{\operatorname{KL}}_{D}(\mathbb{T}\|\rho):=\max_{p^{\prime}\in\mathbb{T}} \operatorname{KL}_{D}(p^{\prime}\|\rho),
\]
\[
\underline{\operatorname{KL}}_{D}(\mathbb{T}\|\rho):=\min_{p^{\prime}\in\mathbb{T}}\operatorname{KL}_{D}(p^{\prime}\|\rho),
\]
describe the extreme KL divergence values of distributions satisfying the IMC $\mathbb{T}$ w.r.t.\ an arbitrary probability distribution $\rho\in\mathbb{P}^{|S|}$, and combining this with abstraction completeness~\eqref{eq:abstraction_completeness}, allows us to conclude that for any hyperrectangular partition $(\mathcal{X}_{i})_{i\in X}$ of state space $\mathcal{X}$, and any probability distribution $\rho\in \mathbb{P}^{|S|}$ --- among which $\rho=p^{\mathrm{u}}$, which is our distribution of interest --- we have
\begin{subequations}
\begin{equation}
    \overline{\operatorname{KL}}_{D}(\mathbb{T}\|\rho)\ge \operatorname{KL}_{D}(p\|\rho) \ge \underline{\operatorname{KL}}_{D}(\mathbb{T}\|\rho).\label{eq:abstraction_bounds}
\end{equation}
 Additionally, that in the limit of $\overline{\lambda}(\mathcal{X}_{i})\to 0$, \eqref{eq:abstraction_bounds} holds with equality, as
\begin{equation}
    \begin{split}
    &\lim_{\overline{\lambda}(\mathcal{X}_{i})\to 0} \overline{\operatorname{KL}}_{D}(\mathbb{T}\|\rho)= \lim_{\overline{\lambda}(\mathcal{X}_{i})\to 0}\operatorname{KL}_{D}(p\|\rho)\\
    &\hspace{4.3cm}= \lim_{\overline{\lambda}(\mathcal{X}_{i})\to 0} \underline{\operatorname{KL}}_{D}(\mathbb{T}\|\rho).\label{eq:abstraction_bounds_convergence}
    \end{split}
\end{equation}\label{eq:abstraction_bounds_general}
\end{subequations}

Note that this result~\eqref{eq:abstraction_bounds_general} enables us to use IMC abstraction~$\mathbb{T}$ to obtain bounds on the KL divergence to uniform of trajectory-distribution discretization~$p$. However, this alone is not enough to characterize their relationship to the KL divergence to uniform of the original trajectory distribution~$T$. In order to relate this result to the system entropy level $\operatorname{KL}(T\|U)$ directly, we require the following lemmas.

\begin{lemma}[Continuous-Discrete Discrepancy] \label{lem:simplification_of_Delta}
    For a distribution density $T$ defined on a connected compact set $\mathcal{S}\subset\mathbb{R}^{n}$ and its discretization $p_{t}:=\int_{\mathcal{S}_{t}}T(s)~\mathrm{d}s$ over partition $\bigcup_{t\in S} \mathcal{S}_{t} = \mathcal{S}$, $\lambda(\mathcal{S}_{t}\cup \mathcal{S}_{t^{\prime}})=0$ for all $t\ne t^{\prime}$, we can write the difference between the KL divergence to uniform of $T$ and the KL divergence to uniform of its discretization $p$, as
    \begin{equation}
    \operatorname{KL}(T\|U)-\operatorname{KL}_{D}(p\|p^{\text{u}})=\operatorname{KL}(T\|T_{D}), \label{eq:discretization_differnce_as_single_KL}
    \end{equation}
    where $T_{D}$ is the piece-wise continuous representation of $p$, defined as 
    \[
    T_D(s) = \sum_{t\in S} \mathbbm{1}_{\mathcal{S}_{t}}(s)\frac{p_{t}}{\lambda(\mathcal{S}_{t})}, \ \ \ \ \mathbbm{1}_{\mathcal{S}_{t}}(s) = 
    \begin{cases} 
        1, & \text{if } s \in \mathcal{S}_{t}, \\
        0, & \text{otherwise,}
    \end{cases}
    \]
    where \(\mathbbm{1}_{\mathcal{S}_{t}}(s)\) is thus an indicator function.
\end{lemma}

\textit{Proof:} See Appendix~\ref{apx:simplification_of_Delta}. \qed

While Lemma~\ref{lem:simplification_of_Delta}, when combined with~\eqref{eq:abstraction_bounds_general}, is insufficient to \emph{upper-bound} the value of $\operatorname{KL}(T\|U)$ through abstraction $\mathbb{T}$, as we neither have access to $T$, nor to $p$, it is sufficient to guarantee a \emph{lower bound}.

\begin{theorem}[Entropy Abstraction Lower Bound]\label{thm:lower_bound}
    \[
    \operatorname{KL}(T\|U)\ge \operatorname{KL}_{D}(p\|p^{\mathrm{u}}) \ge \underline{\operatorname{KL}}_{D}(\mathbb{T}\|p^{\mathrm{u}}),
    \]    
    and
    \[
    \operatorname{KL}(T\|U)=\lim_{\overline{\lambda}(\mathcal{X}_{i})\to 0} \operatorname{KL}_{D}(p\|p^{\mathrm{u}})=\lim_{\overline{\lambda}(\mathcal{X}_{i})\to 0}    \underline{\operatorname{KL}}_{D}(\mathbb{T}\|p^{\mathrm{u}}).
    \]
\end{theorem}

\textit{Proof:} The first (in)equality is found by combining the result in Lemma~\ref{lem:simplification_of_Delta} with the non-negativity of the KL divergence on its right-hand side, while the second follows from~\eqref{eq:abstraction_bounds_general}, see Fig.~\ref{fig:flow_chart_bounds_LB}. \qed

\begin{figure}[tb]
    \centering

    \def\svgwidth{0.495\textwidth}
    \input{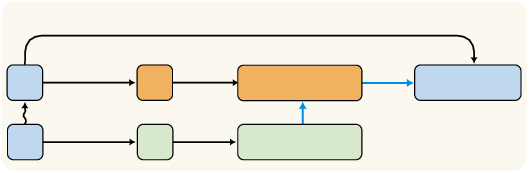}

    \caption{The objects that form a formal \textit{lower bound} on the KL divergence to uniform of system $\mathcal{M}$, i.e.,
    $\operatorname{KL}(T\| U)$, our metric total system entropy.}
    \label{fig:flow_chart_bounds_LB}
\end{figure}

We note that the value of $\underline{\operatorname{KL}}_{D}(\mathbb{T}\|p^{\mathrm{u}})$ and $\overline{\operatorname{KL}}_{D}(\mathbb{T}\|p^{\mathrm{u}})$ can be obtained through standard IMDP methods~\cite{Givan2000BoundedparameterMD}, when combined with the recursive nature of the KL divergence to uniform, which is straightforwardly obtained as an extension of the method in~\cite{vanZutphen2024PredictableIMDPs}, see Lemma~\ref{lem:recursive_dicrete_KL_to_uniform} in the appendix. 

Obtaining an abstraction-based \emph{upper bound} on the relative system entropy requires additional insight. We first present a general lemma that bounds from above the discrepancy between the KL divergence to uniform of a discretization $p$, w.r.t.\ that of its original continuous distribution $T$. As the lemma is on general continuous space distributions and their discretizations (and thus does not rely on the Markov property, or on abstractions), we believe it can prove useful in more general information-theoretic contexts. Besides the requirement that Assumption~\ref{assumptions} holds, we derive an analytic bound on this difference without requiring any additional information on the actual distribution density $T$. 

As each trajectory-space discretization bin $\mathcal{S}_t$ is hyperrectangular, we may describe them through their component intervals, as
\begin{equation}
    \mathcal{S}_t = [\alpha_1^{(t)},\beta_1^{(t)}]\times[\alpha_2^{(t)},\beta_2^{(t)}]\times\cdots\times [\alpha_n^{(t)},\beta_n^{(t)}],\label{eq:mathcalS_t_is_hyperrectangular}
\end{equation}
for all $t\in S$. Let the length of dimension $j\in\{1,2,\dots,n\}$ of element $\mathcal{S}_t$ for $t\in S$ be denoted $\delta_{j}^{(t)}:=\beta_{j}^{(t)}-\alpha_{j}^{(t)}$.

\begin{lemma}[Discretization Difference Upper-Bound]\label{lem:general_delta_upper_bound}
    For a probability distribution $T$ defined on a connected compact set $\mathcal{S}\subset\mathbb{R}^{n}$, subject to Assumption~\ref{assumptions}, and its discretization $p$~\eqref{eq:discretization_of_T}, over hyperrectangular partition $(\mathcal{S}_{t})_{t\in S}$, the difference between the KL divergence to uniform of $T$ and the KL divergence to uniform of its discretization $p$ is upper bounded by
    \[
    \operatorname{KL}(T\|U)-\operatorname{KL}_{D}(p\|p^{\mathrm{u}})\le\varepsilon(p,n,L,\overline{\delta}),
    \]
    where
    \begin{equation}
        \varepsilon(p,n,L,\overline{\delta}):=\sum_{t\in \{1,2,\dots,|p|\}}p_{t}\log\left(1+\frac{n}{2p_{t}}L\overline{\delta}^{n+1}\right),
    \end{equation}
    and $\overline{\delta}:=\max_{t\in S}\max_{j\in\{1,2,\dots,n\}}\delta_j^{(t)}$.
\end{lemma}

\textit{Proof:} See Appendix~\ref{apx:varepsilon_general}.\qed

Next, we specialize the result in Lemma~\ref{lem:general_delta_upper_bound} toward the construction of two distinct upper-bounding approaches. Firstly, a ``global'' approach, that requires minimal information on $\mathcal{M}$, and can be applied a-posteriori to correct more traditional abstraction methods~\cite{Givan2000BoundedparameterMD,vanZutphen2024PredictableIMDPs}, and recover formal upper-bound guarantees. Secondly, a more integrated approach that exploits the result from Lemma~\ref{lem:general_delta_upper_bound} in a ``local'' way, correcting the classic abstraction-method entropy recursion~\cite{vanZutphen2024PredictableIMDPs} at every time-step, thereby making more complete use of information from $\mathcal{M}$, and generally yielding a less conservative upper bound. We present these results as two complete abstraction-based IMC entropy upper-bounding theorems, as Thms.~\ref{thm:global_varepsilon} and~\ref{thm:local_varepsilon} (see Fig.~\ref{fig:flow_chart_bounds_UB}). 

\begin{theorem}[Global Discretization Difference Bound]\label{thm:global_varepsilon}
Under Assumption~\ref{assumptions}, we have
\[
\begin{aligned}
    &\operatorname{KL}(T\|U)&&\le&&\operatorname{KL}_{D}(p\|p^{\mathrm{u}})+\varepsilon(\mathbbm{1}|S|^{-1}\hspace{-0.5mm},n,L,\overline{\delta}) \\
    &&&&&\hspace{7mm}\le\hspace{3.4mm} \overline{\operatorname{KL}}_{D}(\mathbb{T}\|p^{\mathrm{u}})+\varepsilon(\mathbbm{1}|S|^{-1}\hspace{-0.5mm},n,L,\overline{\delta}),\\
\end{aligned}
\]
where $\mathbbm{1}|S|^{-1} := \begin{bmatrix}|S|^{-1} & \cdots & |S|^{-1}\end{bmatrix}^{\top}$, and thus
\[
\varepsilon(\mathbbm{1}|S|^{-1},n,L,\overline{\delta})=\log\left(1+\frac{n}{2}|S|L\overline{\delta}^{n+1}\right).
\]
This upper-bound on the discretization discrepancy further converges to zero for any discretization scheme satisfying 
\begin{equation}
    |S|\overline{\delta}^{n+1}\to0, \label{eq:convergence_requirement}
\end{equation}
yielding
\[
\lim_{\overline{\lambda}(\mathcal{X}_{i})\to 0}\overline{\operatorname{KL}}_{D}(\mathbb{T}\|p^{\mathrm{u}})+ \varepsilon(\mathbbm{1}|S|^{-1},n,L,\overline{\delta})=\operatorname{KL}(T\|U).
\]    
\end{theorem}

\textit{Proof:} See Appendix~\ref{apx:global_varepsilon}, where we derive both bound $\varepsilon(\mathbbm{1}|S|^{-1},n,L,\overline{\delta})$, and the condition on its convergence. \qed

\begin{remark}
The requirement for convergence~\eqref{eq:convergence_requirement} is not very restrictive, as it can be shown to be straightforwardly satisfied by, e.g., the simple grid-based discretization scheme that divides each dimension of $\mathcal{X}$ uniformly into $N$ regions, when $N\to\infty$. 
Let us confirm this by noting that under the aforementioned scheme, $|S|=|X|^{K+1}=N^{n_{x}(K+1)}$, $\overline{\delta}=\overline{\delta}_{x}/N$, where $\overline{\delta}_{x}=\max_{i\in\{1,2,\dots,n_{x}\}}b_i-a_i$, and lastly, $n=(K+1)n_{x}$, which can be substituted into the requirement~\eqref{eq:convergence_requirement} to yield $|S|\overline{\delta}^{n+1}=\overline{\delta}_{x}^{n_{x}(K+1)+1}/N$, which will always go to zero for $N\to\infty$, as $\overline{\delta}_{x}^{n_{x}(K+1)+1}$ is a constant.
\end{remark}

Let $\Phi : \mathbb{P}^{|X|} \times \mathbb{R}^{|X|}\to \mathbb{R}$, be defined as
\begin{equation}
\Phi(p^{\prime},V):= \sum_{j\in X}p^{\prime}_j\log \frac{p^{\prime}_j\lambda(\mathcal{X})}{\lambda(\mathcal{X}_j)}+ \sum_{j\in X}p^{\prime}_j V_j, \label{eq:Phi}
\end{equation}
and let $\Phi^{\varepsilon} : \mathbb{P}^{|X|} \times \mathbb{R}^{|X|}\to \mathbb{R}$, further be defined as 
\[
\begin{split}
\Phi^{\varepsilon}(p^{\prime},V):=\Phi(p^{\prime},V)+\varepsilon(p^{\prime},n_{x},L_{\nabla q},\overline{\delta}).
\end{split}
\]

\begin{theorem}[Local Discretization Difference Bound]\label{thm:local_varepsilon} Under Assumption~\ref{assumptions}, the continuous-trajectory KL divergence to uniform is bounded from above by
\[
\operatorname{KL}(T\|U)\le \operatorname{KL}_{D}^{\varepsilon}(\mathbb{T}\|p^{\mathrm{u}}),
\]   
where
\begin{subequations}\label{eq:local_epsilon_recursion}
\begin{equation}
\operatorname{KL}_{D}^{\varepsilon}(\mathbb{T}\|p^{\mathrm{u}}):= \Phi^{\varepsilon}(\pi,V_{0}^{\varepsilon}), \label{eq:local_epsilon_recursion_0}
\end{equation}
and, for $k=K-1,K-2,\dots,0$,
\begin{equation}
V_k^{\varepsilon}(i)=\max_{P_{k}(i)\in\mathbb{P}_{i}}\Phi^{\varepsilon}(P_{k}(i),V_{k+1}^{\varepsilon}),\label{eq:local_epsilon_recursion_k}
\end{equation}
\end{subequations}
initialized by $V_{K}^{\varepsilon}(i)=0$, for all $i\in X$. Furthermore, for grid-based discretizations satisfying $|S|\overline{\delta}^{n+1}\to 0$, we obtain convergence, as
\[
\lim_{\overline{\lambda}(\mathcal{X}_{i})\to 0}\operatorname{KL}_{D}^{\varepsilon}(\mathbb{T}\|p^{\mathrm{u}})=\operatorname{KL}(T\|U).
\]    
\end{theorem}

\textit{Proof:} See Appendix~\ref{apx:local_varepsilon}.\qed

\begin{figure}[tb]
    \centering

    \def\svgwidth{0.495\textwidth}
    \input{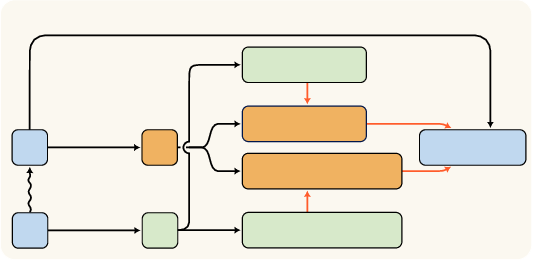}

    \caption{The objects that together form formal \textit{upper bounds} on the KL divergence to uniform of system $\mathcal{M}$, i.e., $\operatorname{KL}(T\| U)$, our metric total system entropy.}
    \label{fig:flow_chart_bounds_UB}
\end{figure}

\subsection{Considerations and algorithms}\label{sec:practical_tools}

The result in Thms.~\ref{thm:global_varepsilon} and~\ref{thm:local_varepsilon}, combined with the recursive algorithm for the KL divergence to uniform, see Lemma~\ref{lem:recursive_dicrete_KL_to_uniform} in the appendix, provides two practical approaches to obtain bounds on the KL divergence to uniform of $T$. Algorithm~\ref{alg:compute_bounds} computes both, yielding the values of bounds (i) $\underline{\operatorname{KL}}_{D}(\mathbb{T}\|U)$ ($\le \operatorname{KL}(T\|U)$), (ii) $\overline{\operatorname{KL}}_{D}(\mathbb{T}\|U)+\varepsilon(\mathbbm{1}|S|^{-1},n,L,\overline{\delta})$ ($\ge \operatorname{KL}(T\|U)$), and (iii) $\operatorname{KL}_{D}^{\varepsilon}(\mathbb{T}\|U)$ ($\ge \operatorname{KL}(T\|U)$). 

\begin{algorithm}[tb]
\caption{Obtaining the upper- and lower bounds}\label{alg:compute_bounds}
    Input: $L,L_{\nabla q},X,(\mathcal{X}_{i})_{i\in X},K,\pi,\underline{P},\overline{P}$.

    $n_{x}\Leftarrow \operatorname{dim}(\mathcal{X}_0)$,
    
    $\overline{\delta} \Leftarrow \max_{i\in X}\max_{j\in\{1,2,\dots,n_{x}\}}b_j^{(i)}-a_j^{(i)}$,
    
    $|S| \Leftarrow |X|^{K+1}$,  
    
    $\varepsilon \Leftarrow \log(1+\frac{1}{2}n|S|L\overline{\delta}^{n+1})$,

    $\mathbb{P}_i\Leftarrow \{p^{\prime}\in\mathbb{P}^{|X|} : \underline{P}_{ij}\le p^{\prime}_j\le\overline{P}_{ij}\}$, $\forall i\in X$.

    Value iterate to find upper- and lower-bounds:

    $\underline{V}_{K}(i)=\overline{V}_{K}(i)=V^{\varepsilon}_{K}(i)\Leftarrow 0, \ \forall i\in X$,
    
    \textbf{for} $k\in\{K-1,K-2,\dots,0\}$ \textbf{do}
    
        \hspace{5mm}\textbf{for} $i\in X$ \textbf{do} 
        
            \vspace{-4mm}
            \begin{align*}
            \underline{V}_k(i) &\Leftarrow \min_{p^{\prime}\in\mathbb{P}_i}\Phi(p^{\prime},\underline{V}_{k+1}),\hspace{3cm}\\ 
            \overline{V}_k(i) &\Leftarrow \max_{p^{\prime}\in\mathbb{P}_i}\Phi(p^{\prime},\overline{V}_{k+1}),\\ 
            V^{\varepsilon}_k(i) &\Leftarrow \max_{p^{\prime}\in\mathbb{P}_i}\Phi^{\varepsilon}(p^{\prime},V^{\varepsilon}_{k+1}),
            \end{align*}
            \vspace{-5mm}

        \hspace{5mm}\textbf{end}
        
    \textbf{end}
    
    Combine and return:
    
    $\underline{\operatorname{KL}}_{D}(\mathbb{T}\|p^{\mathrm{u}}) \Leftarrow \pi^{\top}\underline{V}_0$,
    
    $\overline{\operatorname{KL}}_{D}(\mathbb{T}\|p^{\mathrm{u}})+\varepsilon(\mathbbm{1}|S|^{-1},n,L,\overline{\delta}) \Leftarrow \pi^{\top}\overline{V}_0 + \varepsilon$,
    
    $\operatorname{KL}_{D}^{\varepsilon}(\mathbb{T}\|p^{\mathrm{u}}) \Leftarrow \pi^{\top}V^{\varepsilon}_0$.
\end{algorithm}

Lastly, Lemma~\ref{lem:bounding_the_gradient} aids in finding a bound $L$ on the gradient of the transition distribution density $T$ when having access to initial state distribution $q_{0}$ and the Markov kernel $q$.

\begin{lemma}[Bounding the magnitude of the gradient $T$]\label{lem:bounding_the_gradient}
    Given Assumption~\ref{assumptions} and a constant $L_{q}\in\mathbb{R}_{\ge 0 }$, such that 
    \begin{equation}\label{eq:bounds_on_markov_kernels}
        \|q\|_{\infty}\le L_{q}, \quad \|q_0\|_{\infty}\le L_{q},
    \end{equation}
    then 
    \[
    \|\nabla T\|_{\infty}\le 2(L_{q})^{K}L_{\nabla q}=:L.
    \]
\end{lemma}

\textit{Proof:} See Appendix~\ref{apx:bounding_the_gradient_proof}. \qed

\section{Formal Entropy-Regularized Control} \label{sec:generalizing_to_control} 

Using the theory developed above for formal entropy abstractions, we are now able to construct a formal-abstractions \emph{controller-synthesis} method. Here, we first introduce Markov decision process (MDP) notation and then present a generalized abstraction-based controller-synthesis method that enables the trading off of system entropy (quantified by the KL divergence to uniform of its trajectory distribution) with cumulative cost, with formal guarantees carrying over to the original continuous-state system.

\subsection{Basic preliminaries}

Let a continuous-state MDP be parametrized as
\[
\mathcal{D}=(\mathcal{X},\tilde{U},q_0,[q^{u}]_{u\in \tilde{U}},g,K),
\]
with $\mathcal{X}$, $q_0$, and $K$ as above, $\tilde{U}$ a finite set of actions, $[q^{u}]_{u\in \tilde{U}}$ a finite set of Markov kernels indexed by $u\in \tilde{U}$, and a stage-cost function $g:\mathcal{X}\times \tilde{U}\rightarrow\mathbb{R}$ (time dependence of $g$ is omitted without loss of generality). Let a policy on the discretized state-space $X$ be a function $\mu : \{0,1,\dots,K-1\}\times X\to \tilde{U}$. Let the trajectory distribution density $T^{\mu}$ induced by $\mathcal{D}$ subject to a policy $\mu$ be defined as
\[
T^{\mu}(s) = q_0(x_0)\prod_{k=0}^{K-1}q^{\mu_{k}(\mathfrak{i}(x_{k}))}(x_{k},x_{k+1}),
\]
where $s=x_0,\dots,x_{K}$ is a trajectory in $\mathcal{S}:=\mathcal{X}^{K+1}$, and $\mathfrak{i}(x_{k})$ returns the discrete state $i\in X$ such that $x_{k}\in\mathcal{X}_{i}$.

We aim to minimize a standard expected cumulative stage cost, regularized by trajectory entropy, as quantified by the KL divergence to uniform of $T^{\mu}$. We consider a positive regularization term here, thus corresponding to a \textit{penalization of predictable behavior} when minimizing the cost function, but the method generalizes straightforwardly to predictability encouragement. Specifically, we consider the finite-horizon entropy-regularized objective
\begin{equation}
\underbrace{\mathbb{E}\left[\sum_{k=0}^{K-1}g(x_{k},\mu_{k}(\mathfrak{i}(x_{k})))\right]}_{J^{\mu}} + \operatorname{KL}(T^{\mu}\|U).
\label{eq:entropy_regularized_objective}
\end{equation}
An explicit regularization weight factor has been omitted for notational simplicity.

We model the abstraction as a finite-state MDP on state space $X$, action space $\tilde{U}$, with uncertain transition probabilities; commonly referred to as an interval MDP (IMDP). For each state-action pair $(i,u)\in X\times \tilde{U}$, the transition kernel is constrained to lie in
\begin{equation}
P(i,u)\in\mathbb{P}_{i}(u):=\{p^{\prime}\in\mathbb{P}^{|X|}\mid \underline{P}_{ij}(u)\le p^{\prime}_{j}\le \overline{P}_{ij}(u)\},
\label{eq:mathbbPi_u}
\end{equation}
where
\begin{equation}
\begin{split}
&\underline{P}_{ij}(u) = \min_{x\in \mathcal{X}_i}\int_{\mathcal{X}_j}q^{u}(x,x^{\prime})~\mathrm{d}x^{\prime},\\
&\overline{P}_{ij}(u) = \max_{x\in \mathcal{X}_i}\int_{\mathcal{X}_j}q^{u}(x,x^{\prime})~\mathrm{d}x^{\prime}, 
\label{eq:interval_transitions_control}
\end{split}
\end{equation}
for all $i,j\in X$, $u\in\tilde{U}$. We further define conservative abstract stage costs~\cite{Givan2000BoundedparameterMD}, as
\begin{equation}
\overline{g}(i,u)=\max_{x\in\mathcal{X}_{i}}g(x,u), \ \ \ \ \underline{g}(i,u)=\min_{x\in\mathcal{X}_{i}}g(x,u),
\label{eq:abstract_costs}
\end{equation}
for all $i\in X$, and all $u\in \tilde{U}$.

We synthesize policies by optimizing (minimizing) the two upper bounds associated with the techniques from Thms.~\ref{thm:global_varepsilon} and~\ref{thm:local_varepsilon}, respectively, i.e., our approach yields two bound-optimizing policies and their associated upper-bound values. We additionally compute the associated lower bound for each policy, according to the result from Thm.~\ref{thm:lower_bound}, which bound the performance of the actual continuous system, under the same policy.

\subsection{Formal entropy-regularized bound optimization}

We equip our upper-bound entropy recursion framework (Alg.~\ref{alg:compute_bounds}) with standard IMDP robust dynamic programming for stage costs. We compute (i) a policy $\mu$ that minimizes the \emph{globally-corrected} upper bound (cf.\ Thm.~\ref{thm:global_varepsilon}), and (ii) a policy $\mu^{\varepsilon}$ that minimizes the \emph{locally-corrected} upper bound (cf.\ Thm.~\ref{thm:local_varepsilon}). Both are computed by backward dynamic programming on the abstraction, as described in Alg.~\ref{alg:compute_bounding_policies}.

\begin{algorithm}[!t]
\caption{Computing the formal upper-bound minimizing control policies}\label{alg:compute_bounding_policies}
    Input: $L,L_{\nabla q},X,(\mathcal{X}_{i})_{i\in X},K,\pi,\underline{P},\overline{P}$.

    $n_{x}\Leftarrow \operatorname{dim}(\mathcal{X}_0)$,
    
    $\overline{\delta} \Leftarrow \max_{i\in X}\max_{j\in\{1,2,\dots,n_{x}\}}b_j^{(i)}-a_j^{(i)}$,
    
    $|S| \Leftarrow |X|^{K+1}$,  
    
    $\varepsilon \Leftarrow \log(1+\frac{1}{2}n|S|L\overline{\delta}^{n+1})$,

    $\mathbb{P}_i(u)\hspace{-1mm}\Leftarrow \hspace{-1mm}\{p^{\prime}\in\mathbb{P}^{|X|} : \underline{P}_{ij}(u)\le p^{\prime}_j\le\overline{P}_{ij}(u)\}$, $\forall i\in X, u\in \tilde{U}$.

    $\overline{g}(i,u)\Leftarrow \max_{x\in\mathcal{X}_{i}}g(x,u)$, for all $u\in \tilde{U}$, all $i\in X$,
    
    $\underline{g}(i,u)\Leftarrow \min_{x\in\mathcal{X}_{i}}g(x,u)$, for all $u\in \tilde{U}$, all $i\in X$,

    Value iterate to find upper- and lower-bounds:

    $\underline{V}_{K}(i)=\overline{V}_{K}(i)=\underline{V}^{\varepsilon}_{K}(i)=\overline{V}^{\varepsilon}_{K}(i) \Leftarrow 0, \ \forall i\in X$,
    
    \textbf{for} {$k\in\{K-1,K-2,\dots,0\}$} \textbf{do}

        \hspace{5mm} \textbf{for} {$i\in X$} \textbf{do}
        
            \vspace{-4mm}
            \begin{align*}
            \mu_k(i) &\Leftarrow \arg\min_{u\in \tilde{U}} \max_{p^{\prime}\in\mathbb{P}_i(u)}\overline{g}(i,u)+\Phi(p^{\prime},\overline{V}_{k+1}),\\ 
            \mu^{\varepsilon}_k(i) &\Leftarrow \arg\min_{u\in \tilde{U}} \max_{p^{\prime}\in\mathbb{P}_i(u)}\overline{g}(i,u)+\Phi^{\varepsilon}(p^{\prime},V^{\varepsilon}_{k+1}),\\ 
            \underline{V}_k(i) &\Leftarrow \min_{p^{\prime}\in\mathbb{P}_i(\mu_k(i))}\underline{g}(i,\mu_{k}(i))+\Phi(p^{\prime},\underline{V}_{k+1}),\\ 
            \underline{V}^{\varepsilon}_k(i) &\Leftarrow \min_{p^{\prime}\in\mathbb{P}_i(\mu^{\varepsilon}_k(i))}\underline{g}(i,\mu_{k}^{\varepsilon}(i))+\Phi(p^{\prime},\underline{V}^{\varepsilon}_{k+1}),\\ 
            \overline{V}_k(i) &\Leftarrow \max_{p^{\prime}\in\mathbb{P}_i(\mu_k(i))}\overline{g}(i,\mu_{k}(i))+\Phi(p^{\prime},\overline{V}_{k+1}),\\ 
            \overline{V}^{\varepsilon}_k(i) &\Leftarrow \max_{p^{\prime}\in\mathbb{P}_i(\mu^{\varepsilon}_k(i))}\overline{g}(i,\mu_{k}^{\varepsilon}(i))+\Phi^{\varepsilon}(p^{\prime},\overline{V}^{\varepsilon}_{k+1}),\\ 
            \end{align*}
            \vspace{-13mm}

        \hspace{5mm} \textbf{end}

    \textbf{end}

    Combine and return:
    
    $\underline{J}^{\mu}+\underline{\operatorname{KL}}_{D}(\mathbb{T}^{\mu}\|U) \Leftarrow \pi^{\top}\underline{V}_0$,

    $\underline{J}^{\mu^{\varepsilon}}+\underline{\operatorname{KL}}_{D}(\mathbb{T}^{\mu^{\varepsilon}}\|U) \Leftarrow \pi^{\top}\underline{V}^{\varepsilon}_0$,
    
    $\overline{J}^{\mu}+\overline{\operatorname{KL}}_{D}(\mathbb{T}^{\mu}\|U)+\varepsilon(\mathbbm{1}|S|^{-1},n,L,\overline{\delta}) \Leftarrow \pi^{\top}\overline{V}_0 + \varepsilon$,
    
    $\overline{J}^{\mu^{\varepsilon}}+\operatorname{KL}_{D}^{\varepsilon}(\mathbb{T}^{\mu^{\varepsilon}}\|U) \Leftarrow \pi^{\top}\overline{V}^{\varepsilon}_0$.
\end{algorithm}

Modifying the controller synthesis algorithm to minimize the upper bounds (in contrast to its current functioning as a lower-bound maximizer) is straightforwardly achieved by exchanging the appropriate maximizations by minimizations and vice-versa.

\begin{remark}\label{rem:convex_max_over_polytope}
Note that the maximizations over $p^{\prime}\in\mathbb{P}_i(u)$ in Alg.~\ref{alg:compute_bounding_policies} maximize convex objectives over polytopes; a concave problem. This can thus equivalently be formulated as a finite optimization over extreme points, similar to the cost-only IMDP setting~\cite{Givan2000BoundedparameterMD,vanZutphen2024PredictableIMDPs}, guaranteeing globally optimal solutions. In contrast, in the \textit{predictability encouraging} variation of our algorithm, where the KL divergence to uniform has a negative sign, yields a convex problem; allowing efficient numerical solutions.
\end{remark}

The resulting policies thus minimize the guaranteed KL divergence to uniform versus stage-cost upper bounds when applied to the original, continuous-state system, as
\[
\begin{split}
\mu&=\arg\min_{\mu^{\prime}}\Bigg\{\overline{J}^{\mu^{\prime}}+\overline{\operatorname{KL}}_{D}(\mathbb{T}^{\mu^{\prime}}\|p^{\mathrm{u}})+\varepsilon(\mathbbm{1}|S|^{-1},n,L,\overline{\delta})\Bigg\}, \\ 
\mu^{\varepsilon}&=\arg\min_{\mu^{\prime}} \Bigg\{\overline{J}^{\mu^{\prime}}+\operatorname{KL}_{D}^{\varepsilon}(\mathbb{T}^{\mu^{\prime}}\|p^{\mathrm{u}})\Bigg\},
\end{split}
\]
where 
\begin{equation}
\overline{J}^{\mu}:=\mathbb{E}\Bigg[\sum_{k=0}^{K-1}\overline{g}(\mathfrak{i}(x_{k}),\mu_{k}(\mathfrak{i}(x_{k})))\Bigg]. \label{eq:overline_J_mu}
\end{equation}
Moreover, the resulting bounds apply to the original continuous-state system, as
\[
\begin{split}
&\underline{J}^{\mu}+\underline{\operatorname{KL}}_{D}(\mathbb{T}^{\mu}\|U)\le J^{\mu}+\operatorname{KL}(T^{\mu}\|U)  \\
&\hspace{2.5cm}\le\overline{J}^{\mu}+\overline{\operatorname{KL}}_{D}(\mathbb{T}^{\mu}\|U)+\varepsilon(\mathbbm{1}|S|^{-1},n,L,\overline{\delta}),
\end{split}
\]
for optimal policy $\mu$, and
\[
\begin{split}
&\underline{J}^{\mu^{\varepsilon}}+\underline{\operatorname{KL}}_{D}(\mathbb{T}^{\mu^{\varepsilon}}\|U)\le J^{\mu^{\varepsilon}}+\operatorname{KL}(T^{\mu^{\varepsilon}}\|U)  \\
&\hspace{5.1cm}\le\overline{J}^{\mu^{\varepsilon}}+\operatorname{KL}_{D}^{\varepsilon}(\mathbb{T}^{\mu^{\varepsilon}}\|U),
\end{split}
\]
for optimal policy $\mu^{\varepsilon}$, where $\underline{J}^{\mu}$ refers to~\eqref{eq:overline_J_mu} after replacing the stage-cost upper bound $\overline{g}$ by the lower bound $\underline{g}$.

\section{Examples} \label{sec:examples} 

Below, simplified versions of real-world systems are used to demonstrate the behavior of entropy quantification and entropy-regulated controller synthesis through abstractions.

\subsection{Markov Chain Entropy Bounds}

In this section, we perform a numerical study on the convergence behavior of our abstraction-based entropy-bounds method. Alg.~\ref{alg:compute_bounds} is applied to a heavily simplified, though highly relevant example of a multi-dimensional (clipped) Gaussian transition model.

The exact system we aim to investigate is described by
\begin{equation}
x_{k+1}=x_{k}+w_{k}, \qquad w_{k}\sim\mathcal{N}(0,\Sigma), \label{eq:dynamics_Ex1}
\end{equation}
where $x_k\in\mathcal{X}\subset\mathbb{R}^{n_{x}}$, $\Sigma\in\mathbb{R}^{n_{x}\times n_{x}}$, on the compact hyper-rectangular state space
\[
\mathcal{X}=[a_1,b_1]\times[a_2,b_2]\times\cdots\times[a_{n_{x}},b_{n_{x}}],
\]
coupled with a uniformly distributed placement of $x_{k+1}$ back onto $\mathcal{X}$ whenever the outcome of $x_{k}+w_{k}\notin\mathcal{X}$.

Denoting the probability mass that stochastic system~\eqref{eq:dynamics_Ex1} places outside of the state space $\mathcal{X}$ when in state $x\in\mathcal{X}$, by $o(x,\Sigma)=1-\int_{\mathcal{X}} \mathcal{N}(x^{\prime};x,\Sigma)~\mathrm{d}x^{\prime}$, allows us to describe its Markov kernel $q(x,x^{\prime})$ analytically, as
\[
q(x,x^{\prime}) = \begin{cases}\mathcal{N}(x^{\prime};x,\Sigma) + o(x,\Sigma)/\lambda(\mathcal{X}), & \text{for }x^{\prime}\in\mathcal{X}, \\ 
0 & \text{otherwise},
\end{cases}
\]
where $\lambda(\mathcal{X})=\prod_{j=1}^{n_{x}}(b_j-a_j)$ denotes the Lebesgue measure of $\mathcal{X}$. Further assuming a Gaussian initial state distribution parameterized by $\hat{x}_{0}$, $\Sigma_0$, we have 
\[
q_0(x_0)=\mathcal{N}(x_0;\hat{x}_0,\Sigma_0)+o(\hat{x}_{0},\Sigma_{0})/\lambda(\mathcal{X}), \ \ \ \text{for } x_0\in\mathcal{X}.
\]

While for this specific scenario, expressions for $L_{q}$ and $L_{\nabla q}$ can be obtained analytically, we instead opt to obtain these values numerically in our implementation, emphasizing how this step is tractable even for more more exotic distributions.

Applying the entropy abstraction method to this continuous-state MC for $K=4$, $n_{x}=2$, $a_j=0$, $b_j=1$, for $j\in\{1,2\}$, subject to the aforementioned uniform discretization of each dimension of $\mathcal{X}$ into $N$ regions, yields convergent bounds as displayed in Fig.~\ref{fig:abstraction_bounds_vs_N}. Note that for $N\in \{2,3,4\}$, we have executed Alg.~\ref{alg:compute_bounds} both using the exact maximization (over all extreme points of $\mathbb{P}_i(u)$, see Remark~\ref{rem:convex_max_over_polytope}) and a more efficient numerical alternative. As the number of extreme points in polytopes $\mathbb{P}_{i}(u)$ scales exponentially with $N$, this approach becomes intractable for larger $N$. As the numerical solution agrees with the exact solution at $N\in\{2,3,4\}$, we trust the accuracy of its results for $N\ge 5$.

\begin{figure}[tb]
    \centering
    \includegraphics[width=0.5\linewidth]{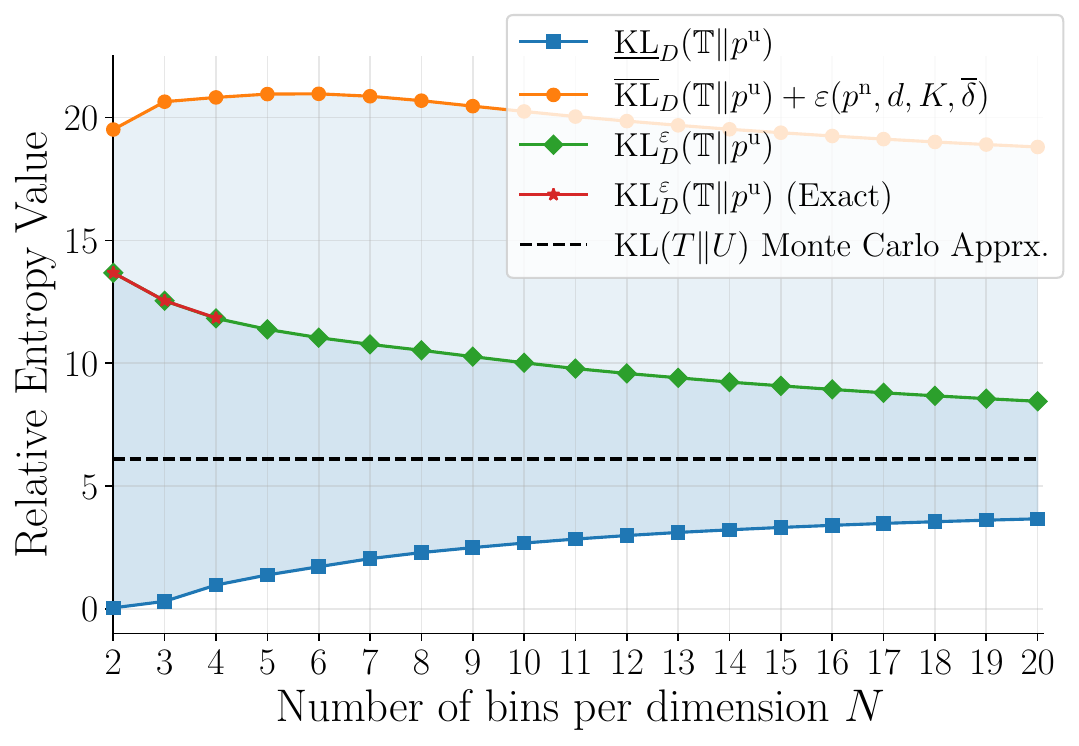}
    \caption{KL divergence to uniform abstraction bounds as a function of discretization resolution as expressed in the number $N$ of equal subdivisions for each dimension.}
    \label{fig:abstraction_bounds_vs_N}
\end{figure}


\subsection{Markov Decision Process Controller Synthesis}

In this section, we balance a stage-cost minimization with system entropy minimization in an autonomous driving example where both speed and predictable behavior are desired. Flipping the regularization term sign, predictability can be encouraged by Alg.~\ref{alg:compute_bounding_policies}. As minimizing the upper bound on the negative KL to uniform regularization term (maximizing predictability) is equivalent to maximizing the lower bound on the original term, the additional mechanics introduced in Thms.~\ref{thm:global_varepsilon} and~\ref{thm:local_varepsilon} are not involved in policy optimization. Only if we are interested in obtaining a performance lower bound on the resulting policy, do we require their application.

We consider a simplified minimum time versus minimum entropy rough terrain descent, where higher velocities are associated with more unpredictable behavior. Let state $v_{k}\in \mathcal{X}:=[0,1]$ describe the velocity of the vehicle, input $u_{k}\in U:=\{0,5,10,15,20\}$ describe an additive acceleration, and let the velocity-dependent stochasticity be modeled as another additive (stochastic) acceleration $w_{k}\in[-1,1]$, for times $k\in\{0,1,\dots,K\}$. Stage costs $g(v_k,u_k)=-v_k$ decrease with increased velocity, while the level of unpredictability increases with increased velocity, as the system dynamics follow
\[
v_{k+1} = 0.8v_{k} + 0.01u_{k} + w_{k}, \ \ \ w_{k}\sim d(~\cdot~  |v_{k}),
\]
where we use the three left-, middle-, and right points
\[
\begin{bmatrix}
    l(v) \\
    m(v) \\
    r(v)
\end{bmatrix}=(1-v)\begin{bmatrix}
    0 \\
    0.05 \\
    0.1
\end{bmatrix} + v\begin{bmatrix}
    -0.8 \\
    -0.2 \\
    0
\end{bmatrix},
\]
to define the triangular probability density of $w_{k}$ as a function of $v_{k}$, as
\[
d(\omega|\nu)=\frac{2}{r(\nu)-l(\nu)}\begin{cases}
    \frac{\omega - l(\nu)}{m(\nu)-l(\nu)}, &\text{if }\omega\in[l(\nu),m(\nu)], \\
    \frac{r(\nu)-\omega}{r(\nu)-m(\nu)}, &\text{if }\omega\in[m(\nu),r(\nu)], \\
    0, & \text{otherwise},
    \end{cases}
\]
such that
\[
\operatorname{Prob}(w_{k}\in A |v_{k}=\nu) = \int_{A}d(\omega|\nu)~\mathrm{d}\omega.
\]

The functions $(l,m,r)$ widen the support of $w_k$ and cause it to shift downwards as velocity increases, modeling more unpredictable disturbances at higher speeds while ensuring the closed-loop dynamics satisfy $v_{k+1}\in[0,1]$ almost surely for all $v_k\in[0,1]$, and all $u_k\in U$. By minimizing the upper bound on the entropy-regularized control performance, we obtain policies that are guaranteed to meet or exceed the associated level of performance, ensuring that the controlled system meets certain safety specs for example.

We partition the state space $\mathcal{X}$ uniformly into $N=80$ regions to yield an IMDP abstraction according to~\eqref{eq:interval_transitions_control}. Solving the resulting robust dynamic program (Alg.~\ref{alg:compute_bounding_policies}) produces a policy that trades off speed against unpredictability. To illustrate the effect of regularization on the system trajectories, we scale cost function $g$ by $\varphi\in\{2.3,2.56\}$, and also compute a purely minimum-time policy $\mu_{\text{DP}}$ (using standard IMDP stage cost bound minimization~\cite{Givan2000BoundedparameterMD} without entropy regularization).

The simulation results are visualized in Fig.~\ref{fig:entropy_bound_minimizing_control}, where the variance of the trajectory under $\mu(\varphi=2.56)$ is seemingly larger than the others, but this is due to a rough bifurcation of trajectories, either going up or down all the way. The associated KL divergence to uniform of the three policies is found though Monte Carlo simulation as
\[
\begin{aligned}
& \operatorname{KL}(T^{\mu_{\text{DP}}}\|U) =  19.2, \\ 
& \operatorname{KL}(T^{\mu(\varphi=2.56)}\|U) =  25.3, \\ 
& \operatorname{KL}(T^{\mu(\varphi=2.3)}\|U) =  32.1, 
\end{aligned}
\]
showing the clear effect of increasing the relative level of entropy regularization. Increased emphasis on entropy minimization results in an avoidance of the high-velocity regime to prevent increased unpredictability, instead favoring moderate speeds where the disturbance distribution is narrower. In contrast, the minimum-time controller purely drives the system toward high velocities, resulting in significantly higher trajectory entropy.

The bounds we obtain are further 
\[
-36.3 \le J^{\mu(\varphi=2.56)}-\operatorname{KL}(T^{\mu(\varphi=2.56)}\|U) \le -34.8,
\]
when we scale $g$ by $\varphi = 2.56$, and
\[
-36.2\le 
J^{\mu(\varphi=2.3)}-\operatorname{KL}(T^{\mu(\varphi=2.3)}\|U)\le  -34.3,
\]
when $g$ is scaled by $\varphi = 2.3$, where $J^{\mu}$, for any policy $\mu$, is defined as~\eqref{eq:entropy_regularized_objective}. Both approaches thus achieve a bound gap of only 5\% of the total objective in this example.

\begin{figure}[tb]
    \centering
    \includegraphics[width=0.5\linewidth]{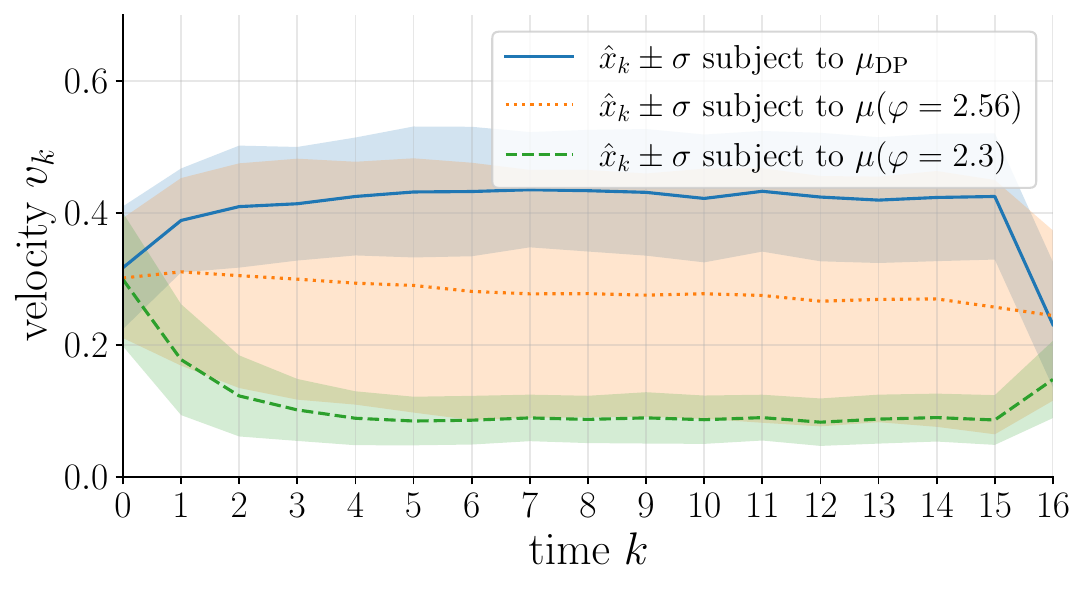}
    \caption{Behavior of the bumpy-hill AV example system under the unregularized policy $\mu_{\text{DP}}$, the entropy-regularized policy with global $\varepsilon$ correction $\mu$, and the locally-corrected entropy-regularized policy $\mu^{\varepsilon}$.}
    \label{fig:entropy_bound_minimizing_control}
\end{figure}



\section{Conclusions and Discussion} \label{sec:conc_and_discussion}

In this paper, we develop a formal abstraction framework for analyzing and controlling trajectory entropy in continuous-state stochastic systems. We derive computable upper and lower bounds on the KL divergence to uniform of trajectory distributions using interval MDP abstractions, all of which convergence to the true KL divergence to uniform as the discretization resolution is increased. The resulting theory enables formal entropy-aware controller synthesis that trades predictability against control performance while preserving formal guarantees for the original continuous system. Numerical studies confirm convergence of the proposed bounds with discretization refinement, and that the synthesized policies are able to effectively regulate entropy in practice. Future work will focus on further reducing conservatism of abstraction-based bounds, and extending the approach to the infinite horizon setting, richer specifications, and learning-based system models.

\bibliographystyle{IEEEtran}
\bibliography{mybibliography}

\appendix

\part*{Appendix}

\section{Entropy abstraction lower bound (Lemma~\ref{lem:simplification_of_Delta})}\label{apx:simplification_of_Delta}

Writing the KL divergence to uniform of $T$~\eqref{eq:T_to_uniform} as a sum over the intervals $\mathcal{S}_{t}$, as
\[
\operatorname{KL}(T\|U) = \sum_{t\in S}\int_{\mathcal{S}_{t}}T(s)\log(T(s)\lambda(\mathcal{S}))~\mathrm{d}s,
\]
when combined with the explicit expression for the KL divergence to uniform of $p$~\eqref{eq:p_to_uniform}, and identity~\eqref{eq:discretization_of_T}, allows the discretization difference~\eqref{eq:discretization_differnce_as_single_KL} to be formulated as a single sum over $t\in S$, as
\begin{equation}
\begin{split}
&\operatorname{KL}(T\|U)-\operatorname{KL}_D(p\|p^{\text{u}})=\\
&\hspace{10mm}\sum_{t\in S}\int_{\mathcal{S}_{t}}T(s)\left(\log(T(s)\lambda(\mathcal{S}))-\log\left(\frac{p_{t}\lambda(\mathcal{S})}{\lambda(\mathcal{S}_{t})}\right)\right)~\mathrm{d}s, 
\label{eq:main}
\end{split}
\end{equation}
where we have used the fact that $\log\left(\frac{p_{t}\lambda(\mathcal{S})}{\lambda(\mathcal{S}_{t})}\right)$ is a constant w.r.t.\ $s\in\mathcal{S}_{t}$. By collecting the terms inside the sum of~\eqref{eq:main} into a single logarithmic term, we obtain
\[
\begin{split}
&\operatorname{KL}(T\|U)\hspace{-0.3mm}-\hspace{-0.2mm}\operatorname{KL}_D(p\|p^{\text{u}})=\hspace{-1mm}\sum_{t\in S}\hspace{-0.2mm}\int_{\mathcal{S}_{t}}\hspace{-1.7mm}T(s)\log\left(\frac{T(s)}{p_{t}/\lambda(\mathcal{S}_{t})}\right)\mathrm{d}s\\
&\hspace{65mm} =\operatorname{KL}(T\|T_{D}),
\end{split}
\]
where we have recognized
\[
T_{D}(s)=\begin{cases}
    \frac{p_{t}}{\lambda(\mathcal{S}_{t})}, & \text{when } \ s\in \mathcal{S}_{t}.
\end{cases}
\]\qed

\section{Discretization difference upper-bound (Lemma~\ref{lem:general_delta_upper_bound})}\label{apx:varepsilon_general}

Let us denote the discretization difference by $\Delta$, i.e., using Lemma~\ref{lem:simplification_of_Delta}, we define
\begin{equation}
\begin{split}
&\Delta := \operatorname{KL}(T\|U)-\operatorname{KL}_{D}(p\|p^{\mathrm{u}})\\
&\hspace{2cm}=\sum_{t\in S}\int_{\mathcal{S}_{t}}T(s)\log\left(\frac{T(s)}{p_{t}/\lambda(\mathcal{S}_{t})}\right)~\mathrm{d}s.\label{eq:Delta}
\end{split}
\end{equation}
Our goal is to find a computable upper-bound $\varepsilon$, such that $\varepsilon\ge \Delta$, without explicit knowledge of the underlying distribution density $T$. We assume only that $T$ a continuous distribution density with $\|\nabla T\|_{\infty}\le L$, and that it induces discrete distribution $p$ according to~\eqref{eq:discretization_of_T}. In other words, $\varepsilon$ must satisfy
\begin{subequations}\label{eq:prob_statement}
\begin{align}
\varepsilon\ge\max_{\tau}\sum_{t\in S}&\int_{\mathcal{S}_{t}}\tau(s)\log\left(\frac{\tau(s)}{p_{t}/\lambda(\mathcal{S}_{t})}\right)~\mathrm{d}s, \label{eq:prob_statement_objective} \\ 
\text{s.t. }&\tau : \mathcal{S}\rightarrow \mathbb{R}_{\ge 0},\label{eq:non_negativity_constraint}\\
    &\int_{\mathcal{S}_{t}}\tau(s)~\mathrm{d}s=p_{t}, \forall t \in S, \label{eq:T_must_map_to_p_constraint} \\
    &\|\nabla \tau(s)\|_{\infty}\le L, \  \forall t \in S, \label{eq:lipschitz_constraint} \\
    &\tau \text{ continuous on } \mathcal{S}. \label{eq:global_continuity_constraint}
\end{align}
\end{subequations}

Clearly, due to monotonicity of the logarithm and the sum, the expression 
\begin{equation}
\begin{split}
&\sum_{t\in S}\int_{\mathcal{S}_{t}}\tau(s)\log\left(\frac{\max_{c\in\mathcal{S}_{t}}\{\tau(c)\}}{p_{t}/\lambda(\mathcal{S}_{t})}\right)~\mathrm{d}s =\\
&\hspace{3cm}\sum_{t\in S}p_{t}\log\left(\frac{\max_{c\in\mathcal{S}_{t}}\{\tau(c)\}}{p_{t}/\lambda(\mathcal{S}_{t})}\right),\label{eq:first_upper-bounding_objective}
\end{split}
\end{equation}
forms an upper-bound on the value of objective~\eqref{eq:prob_statement_objective} for any $\tau$ continuous on $\mathcal{S}_{t}$ for $t\in S$, among which is the optimal $\tau^{*}$ that maximizes~\eqref{eq:prob_statement}. Instead of finding the exact value of $\max_{c\in \mathcal{S}_{t}}\{\tau(c)\}$ for every $t\in S$, and optimizing the resulting objective function~\eqref{eq:first_upper-bounding_objective} over $\tau$, we instead formulate analytical upper-bounds on $\max_{c\in \mathcal{S}_{t}}\{\tau(c)\}$ for every individual $t\in S$, given (a relaxed version of) constraints (\ref{eq:non_negativity_constraint}-\ref{eq:global_continuity_constraint}). Substituting those back into~\eqref{eq:first_upper-bounding_objective} guarantees a value that upper-bounds~\eqref{eq:prob_statement}, and therefore can serve as our expression for $\varepsilon$.

\subsection{Upper bounding the maximum tau value}

Toward the upper-bound on $\max_{c\in \mathcal{S}_{t}}\{\tau(c)\}$, for every $t\in S$, and every $\tau$ satisfying (\ref{eq:non_negativity_constraint}-\ref{eq:global_continuity_constraint}), we define $\overline{T}_{t}^{*}$ as the solution to optimization problem
\begin{subequations}\label{eq:max_T_over_St_problem}
\begin{align}
&\overline{T}_{t}^{*}:=\max_{\tau}\max_{c\in\mathcal{S}_{t}}\{\tau(c)\}, \\
\text{s.t. }&\tau : \mathcal{S}\to\mathbb{R}, \label{eq:first_tau_constraint_Ttstar}\\
    &\int_{\mathcal{S}_{t}}\tau(s)~\mathrm{d}s=p_{t}, \\
    &\|\nabla \tau(s)\|_{\infty}\le L, \\
    &\tau \text{ continuous on } \mathcal{S}_{t}, \label{eq:last_tau_constraint_Ttstar}
\end{align}
\end{subequations}
where the constraint relaxations w.r.t.\ the original (\ref{eq:non_negativity_constraint}-\ref{eq:global_continuity_constraint}) consist of (i) omitting the non-negativity constraint on $\tau$, and (ii) replacing the global continuity constraint by a local one. 

We note that~\eqref{eq:max_T_over_St_problem} can be rewritten as
\begin{equation}
\overline{T}_{t}^{*}:=\max_{c\in \mathcal{S}_{t}} \overline{T}_{t}(c), \label{eq:Tstar_main_problem}
\end{equation}
when $\overline{T}_{t}(c)$ is the solution to the inner-problem for a selected trajectory $c\in\mathcal{S}_{t}$, as
\begin{equation}\label{eq:T(c)}
    \overline{T}_{t}(c):=\max_{\tau} \ \tau(c),  \ \ \   \text{s.t. } (\ref{eq:first_tau_constraint_Ttstar}-\ref{eq:last_tau_constraint_Ttstar}).
\end{equation}
So, instead of searching for the highest achievable value over all points $s\in \mathcal{S}_{t}$, we first ask what is the highest value we can achieve at a selected point $c\in \mathcal{S}_{t}$. 

As mentioned in the main text, as every partition $\mathcal{S}_{t}$ is hyper-rectangular, we may describe it through its component intervals $\mathcal{S}_{t} = [\alpha_1^{(t)},\beta_1^{(t)}]\times[\alpha_2^{(t)},\beta_2^{(t)}]\times\cdots\times [\alpha_n^{(t)},\beta_n^{(t)}]$, for all $t\in S$. 

\begin{lemma}[Derivative of Solution to~\eqref{eq:T(c)}] \label{lem:NablaTR(c)}
For a function $\tau^{*}_c$ to maximize~\eqref{eq:T(c)}, for a selected $c\in \mathcal{S}_{t}$, it must satisfy
\begin{equation}
\nabla_j \tau^{*}_c(s) = \begin{cases}
    L & \text{for } s_j\in[\alpha_j^{(t)},c_j), \\
    -L & \text{for } s_j\in[c_j,\beta_j^{(t)}],
\end{cases} \label{eq:candidate_derivative_T(s)}
\end{equation}
for $j\in\{1,2,\dots,n\}$, for all $s\in \mathcal{S}_{t}$, where the choice to let $c_j$ belong to the second interval instead of the first is arbitrary and does not affect the following.
\end{lemma}

\textit{Proof:} Let a candidate for optimality $\tau^{*}_c(s)$ satisfy~\eqref{eq:candidate_derivative_T(s)} and all constraints in~\eqref{eq:T(c)}. We can then write all alternative feasible solutions as perturbed versions of our candidate, as $\tau(s):=\tau^{*}_c(s)+\gamma(s)$. For these $\tau(s)$ to remain feasible (satisfy the integral and derivative constraints in~\eqref{eq:T(c)}), perturbation $\gamma(s)$ must satisfy $\int_{s\in \mathcal{S}_{t}}\gamma(s)~\mathrm{d}s=0$, and 
\[
\nabla_j \gamma(s) \in \begin{cases}
    [-2K,0] & \text{for } s_j\in[\alpha_j^{(t)},c_j), \\
    [0,2K] & \text{for } s_j\in[c_j,\beta_j^{(t)}],
\end{cases} 
\]
for $j\in\{1,2,\dots,n\}$.

Note that $\gamma(s)$, for any $s\in \mathcal{S}_{t}$, is restricted to be non-increasing in the $j$-direction when $s_j\in[\alpha_j^{(t)},c_j)$ and non-decreasing when $s_j\in[c_j,\beta_j^{(t)}]$. As a consequence, if we take $\gamma(s)$ at a point $s\in \mathcal{S}_{t}$ and fix all $s_r$ for $r\in\{1,2,\dots,n\}\setminus \{j\}$, we have that the resulting line parametrized by $s_j$ achieves a minimum at $s_j=c_j$, i.e., $\gamma(\begin{bmatrix} s_1 & s_2 & \cdots & s_j & \cdots & s_n\end{bmatrix}^{\top})\ge\gamma(\begin{bmatrix} s_1 & s_2 & \cdots & c_j & \cdots & s_n\end{bmatrix}^{\top})$, for any $s\in \mathcal{S}_{t}$. Clearly, this still holds when we replace more than one element of $s$ by its corresponding element in $c$, which implies that $\gamma(c)\le\gamma(s)$ for all $s\in \mathcal{S}_{t}$, i.e., the minimum of all admissible perturbations $\gamma$ lies at $c$.

Since $\gamma(c)\le \gamma(s)$ for all $s\in \mathcal{S}_{t}$, and $\int_{\mathcal{S}_{t}}\gamma(s)~\mathrm{d}s=0$, we must have that $\gamma(c)\le 0$, with equality only for $\gamma(s)\equiv0$. From the fact that all admissible perturbations $\gamma(s)$ w.r.t.\ $\tau_c^{*}(s)$ for $s\in \mathcal{S}_{t}$ must have $\gamma(c)\le 0$, we conclude that there are no functions $\tau(s)$ that improve upon $\tau_c^{*}(c)$, i.e., $\tau(c)\le \tau_c^{*}(c)$ for all $\tau(s)$ that satisfy~\eqref{eq:T(c)}. \qed

As mentioned in the main text, for any $j\in\{1,2,\dots,n\}$, $t\in S$, the length of the interval defining the $j$-th dimension of $\mathcal{S}_{t}$, is denoted $\delta_{j}^{(t)}:=\beta_j^{(t)}-\alpha_j^{(t)}$.

\begin{lemma}[Solution to~\eqref{eq:T(c)}] \label{lem:TR(c)}
The solution to~\eqref{eq:T(c)}, for any $c\in \mathcal{S}_{t}$, admits the closed-form representation
    \[
    \begin{split}
    &\overline{T}_{t}(c)=\frac{p_{t}}{\lambda(\mathcal{S}_{t})} + \frac{1}{2\lambda(\mathcal{S}_{t})}L\\
    &\hspace{1mm}\cdot\sum_{j=1}^{n}\left(\prod_{\ell\in\{1,\dots,n\}\setminus\{j\}}\hspace{-3mm}\delta_{\ell}^{(t)}\right)\left(\left(c_{j}-\alpha_{j}^{(t)}\right)^{2}+\left(\beta_{j}^{(t)}-c_{j}\right)^{2}\right).
    \end{split}
    \]
\end{lemma}

\textit{Proof} Integration of~\eqref{eq:candidate_derivative_T(s)}, while satisfying continuity, provides us with a form of $\tau^{*}_c$, as
\begin{equation}
\tau_c^{*}(s) = o_{c}(s)^\top (s-c)  + C_c, \label{eq:T_c_general_form}
\end{equation}
where
\[
o_{c,j}(s_j):=\begin{cases}
    L & \text{for } s_j\in[\alpha_j^{(t)},c_j),\\
    -L & \text{for } s_j\in[c_j,\beta_j^{(t)}],
\end{cases}
\]
for all dimensions $j\in\{1,2,\dots,n\}$.

As $\tau_c^{*}(s)$ for $s\in \mathcal{S}_{t}$ is a solution to~\eqref{eq:T(c)}, we have that $\overline{T}_{t}(c)=\tau_c^{*}(c)$ for all $c\in \mathcal{S}_{t}$. From inspection of equation~\eqref{eq:T_c_general_form}, it becomes clear that $\tau_c^{*}(c)=C_c$, i.e., $\overline{T}_{t}(c)=C_c$.

The value of $C_c$ is determined by the integral constraint $p_{t}=\int_{s\in \mathcal{S}_{t}}T(s)~\mathrm{d}s$, which becomes
\begin{equation}
\begin{split}
&p_{t}=\int_{s\in \mathcal{S}_{t}}o_c(s)^{\top}(s-c) + C_c~\mathrm{d}s\\
&\hspace{1cm}=\int_{s\in \mathcal{S}_{t}}o_c(s)^{\top}(s-c)~\mathrm{d}s + \lambda(\mathcal{S}_{t})C_c, \label{eq:p_{T}_{t}mplicit}
\end{split}
\end{equation}
where $\int\limits_{\mathclap{s\in \mathcal{S}_{t}}}o_c(s)^{\top}(s-c)~\mathrm{d}s =\sum_{j=1}^{n}\int\limits_{\mathclap{s\in \mathcal{S}_{t}}}o_{c,j}(s_j)(s_j-c_j)~\mathrm{d}s$, where, for every $j\in \{1,2,\dots,n\}$, we find
\[
\begin{split}
&\int\limits_{\mathclap{s\in \mathcal{S}_{t}}}o_{c,j}(s_j)(s_j-c_j)~\mathrm{d}s\\
&\hspace{2cm}=\Bigg(\prod_{\ell\in\{1,\dots,n\}\setminus\{j\}}\hspace{-4mm}\delta_{\ell}^{(t)}\Bigg)\int\limits_{\mathclap{s_j\in[\alpha_j^{(t)},\beta_j^{(t)}]}}o_{c,j}(s_j)(s_j-c_j)~\mathrm{d}s_j, 
\end{split}
\]
as the function inside the integral only depends on $s_j$ and not on any of the other $s_\ell$ for $\ell\in\{1,2,\dots,n\}\setminus\{j\}$. We solve the above integral to obtain
\[
\begin{split}
&\int_{s_j\in[\alpha_j^{(t)},\beta_j^{(t)}]}o_{c,j}(s_j)(s_j-c_j)~\mathrm{d}s_j\\
&\hspace{2.3cm}=- L\frac{1}{2}\left(\left(c_j-\alpha_j^{(t)}\right)^{2}
+\left(\beta_j^{(t)}-c_j\right)^{2}\right),
\end{split}
\]
for every $j\in\{1,2,\dots,n\}$.

Substituting these results back into~\eqref{eq:p_{T}_{t}mplicit} yields
\[
\begin{split}
&p_{t}=\lambda(\mathcal{S}_{t})C_c -\frac{1}{2}L\\
&\hspace{2mm}\cdot \sum_{j=1}^{n}\Bigg(\prod_{\ell\in\{1,\dots,n\}\setminus\{j\}}\hspace{-4mm}\delta_{\ell}^{(t)}\Bigg)\left(\left(c_{j}-\alpha_{j}^{(t)}\right)^{2}
+\left(\beta_{j}^{(t)}-c_{j}\right)^{2}\right),
\end{split}
\]
which we can re-order to obtain
\[
\begin{split}
&C_c=\frac{p_{t}}{\lambda(\mathcal{S}_{t})} + \frac{1}{2\lambda(\mathcal{S}_{t})}L\\
&\hspace{2mm}\cdot \sum_{j=1}^{n}\Bigg(\prod_{\ell\in\{1,\dots,n\}\setminus\{j\}}\hspace{-4mm}\delta_{\ell}^{(t)}\Bigg)\left(\left(c_{j}-\alpha_{j}^{(t)}\right)^{2}+\left(\beta_{j}^{(t)}-c_{j}\right)^{2}\right),
\end{split}
\]
where we recall $\overline{T}_{t}(c)=\tau^{*}_c(c)=C_c$. \qed

\begin{theorem}[Solution to~\eqref{eq:Tstar_main_problem}]\label{thm:relaxation_solution}
The solution to~\eqref{eq:Tstar_main_problem} takes the form $\overline{T}_{t}^{*}=\frac{p_{t}}{\lambda(\mathcal{S}_{t})} + \frac{1}{2}L\sum_{j=1}^{n}\delta_j^{(t)}$.
\end{theorem}

\textit{Proof} As $\overline{T}_{t}^{*}=\max_{c\in \mathcal{S}_{t}}\overline{T}_{t}(c)$, and we know from Lemma \ref{lem:TR(c)} that
\[
\begin{split}
&\overline{T}_{t}(c) = \frac{p_{t}}{\lambda(\mathcal{S}_{t})} + \frac{1}{2\lambda(\mathcal{S}_{t})}L\\
&\hspace{2mm}\cdot \sum_{j=1}^{n}\Bigg(\prod_{\ell\in\{1,\dots,n\}\setminus\{j\}}\hspace{-4mm}\delta_{\ell}^{(t)}\Bigg)\left(\left(c_{j}-\alpha_{j}^{(t)}\right)^{2} +\left(\beta_{j}^{(t)}-c_{j}\right)^{2}\right),
\end{split}
\]
which, since $(c-\alpha)^{2}+(\beta-c)^{2}$ is a parabola with its minimum at $(\alpha+\beta)/2$, is clearly maximized for $c_j\in\{\alpha_j^{(t)},\beta_j^{(t)}\}$ for all $j\in\{1,2,\dots,n\}$. Note that thus for, e.g., $c_j=\alpha_j^{(t)}$, for all $j\in\{1,2,\dots,n\}$, we find
\[
\begin{split}
&\overline{T}_t^{*}=\max_{c\in \mathcal{S}_{t}}C_c \\
&= \frac{p_{t}}{\lambda(\mathcal{S}_{t})} + \frac{1}{2\lambda(\mathcal{S}_{t})}L\sum_{j=1}^{n}\Bigg(\prod_{\ell\in\{1,\dots,n\}\setminus\{j\}}\hspace{-4mm}\delta_{\ell}^{(t)}\Bigg)\left(\beta_{j}^{(t)}-\alpha_{j}^{(t)}\right)^{2} \\
&= \frac{p_{t}}{\lambda(\mathcal{S}_{t})} + \frac{1}{2\lambda(\mathcal{S}_{t})}L\sum_{j=1}^{n}\delta_j^{(t)}\lambda(\mathcal{S}_{t})= \frac{p_{t}}{\lambda(\mathcal{S}_{t})} + \frac{1}{2}L\sum_{j=1}^{n}\delta_j^{(t)}.
\end{split}
\] \qed

\textit{Proof of Lemma~\ref{lem:general_delta_upper_bound}:} Substitution of the upper-bound $\overline{T}_{t}^{*}\ge\max_{c\in \mathcal{S}_{t}}\{\tau(c)\}$ solution, see Thm.~\ref{thm:relaxation_solution}, back into~\eqref{eq:first_upper-bounding_objective}, we obtain
\[
\begin{split}
&\Delta\le\sum_{t\in S} p_{t}\log\left(\frac{\overline{T}_{t}^{*}}{p_{t}/\lambda(\mathcal{S}_{t})}\right)\\
&\hspace{2mm}=\sum_{t\in S}p_{t}\log\left(1 + \frac{1}{2p_{t}}L\lambda(\mathcal{S}_{t})\sum_{j=1}^{n}\delta_j^{(t)}\right),
\end{split}
\]
which we further bound from above through the use of $\overline{\delta}:=\max_{t\in S}\max_{j\in\{1,2,\dots,n\}}\delta_j^{(t)}$, for which we note that $\max_{t\in S}\lambda(\mathcal{S}_{t})\le (\overline{\delta})^{n}$, to find an even simpler upper bound (which for uniform hypercubic partitioning does not introduce additional conservatism), as
\begin{equation}
\Delta\le \sum_{t\in S}p_{t}\log\left(1 + \frac{n}{2p_{t}}L\overline{\delta}^{n+1}\right).\label{eq:last_component_bound}
\end{equation}
\qed

\section{Global discretization difference bound (Thm.~\ref{thm:global_varepsilon})}\label{apx:global_varepsilon}

\begin{lemma}[Concavity]\label{lem:concavity_f}
Functions of the form
\[
f(x) = \begin{cases}
    x\log(1+Cx^{-1}) &\text{for }x>0,\\
    0 & \text{for } x=0,
\end{cases}
\]
are concave in $x\ge 0$, for $C\ge 0$.
\end{lemma}

\textit{Proof} The second derivative $f^{\prime\prime}(x)= \frac{C}{\ln(2)}\frac{-C}{x(x+C)^2}$, is clearly negative for all $x> 0$, $C\ge 0$. \qed

\textit{Proof of Thm.~\ref{thm:global_varepsilon}} Jensen's inequality states that for any concave function $f(x)$, w.r.t.\ any finite set of points $p_1,p_2, \ldots, p_n$, and nonnegative weights $\alpha_1, \alpha_2, \ldots, \alpha_n$, with $\sum_i \alpha_i=1$, the following inequality holds $f\left(\sum_{t\in S} \alpha_t p_{t}\right) \geq \sum_{t\in S} \alpha_t f\left(p_{t}\right)$.

Using the fact that our discretization of $T$ satisfies $\sum_{t\in S}p_{t}=1$ and $p_{t}\ge 0$ for all $t\in S$, and that~\eqref{eq:last_component_bound} is concave in $p_{t}$ (Lemma \ref{lem:concavity_f}), Jensen's inequality tells us that
\[
f(1/|S|)=f\left(\sum_{t\in S}\frac{1}{|S|} p_{t}\right)\ge\sum_{t\in S}\frac{1}{|S|}f(p_{t}),
\]
which, when multiplied by $|S|$, yields $|S| f\left(1/|S|\right)\ge \sum_{t\in S}f(p_{t})$, for any choice of $(p_{t})_{t\in S}$ that sums to one. We thus conclude that, regardless of the true $p$, the sum over~\eqref{eq:last_component_bound} can be upper bounded as
\[
\begin{split}
\sum_{t\in S}p_{t}\log(1+\frac{nL}{2p_{t}}\overline{\delta}^{n+1})\le |S|\left(1/|S|\log(1+\frac{n}{2}|S|L\overline{\delta}^{n+1})\hspace{-1mm}\right)\hspace{-1mm}.
\end{split}
\]
Further, we may conclude that any discretization scheme satisfying $|S|\overline{\delta}^{n+1}\to0$ ensures $\varepsilon\to 0$, as $\varepsilon\ge \Delta\to 0$ from the fact that $\log(1)=0$, implying convergence for increasingly fine discretizations. \qed

\section{Local discretization difference bound (Thm.~\ref{thm:local_varepsilon})} \label{apx:local_varepsilon}

The KL divergence to uniform possesses the well-known recursive property of entropy metrics on Markov systems~\cite{Savas2020EntropyMDP_TL,Chen2022EntropyRM_LTL,Srivastava2021ParameterizedMDPs,vanZutphen2024PredictableIMDPs}. Letting $\Phi$ as~\eqref{eq:Phi}, and letting $\mathbb{R}^{\mathcal{X}}$ denote the space of functions mapping values in $\mathcal{X}$ to the real number line, we define functional $\phi : \mathcal{P}(\mathcal{S}) \times \mathbb{R}^{\mathcal{X}}\to \mathbb{R}$, as
\[
\phi(q,\mathcal{V}):=\int_{\mathcal{X}}q(x^{\prime})\log \frac{q(x^{\prime})}{1/\lambda(\mathcal{X})}~\mathrm{d}x^{\prime}+\int_{\mathcal{X}}q(x^{\prime})\mathcal{V}(x^{\prime})~\mathrm{d}x^{\prime}. 
\]

\begin{lemma}[Discrete KL Divergence to Uniform Recursivity]
\label{lem:recursive_dicrete_KL_to_uniform}
Let $p\in\mathbb{P}^{|S|}$ be a discrete trajectory distribution on $S=X^{K+1}$, induced by an initial distribution $\pi$ and Markov kernels $\{P_k\}_{k=0}^{K-1}$. The discrete KL divergence to uniform of $p$ then admits the form
\begin{subequations}\label{eq:discrete_KL_recursion}
\begin{equation}
\operatorname{KL}_{D}(p\|p^{\mathrm{u}})=\Phi(\pi,V_{0}), \label{eq:discrete_KL_recursion_0}
\end{equation}
where, for $k=K-1,K-2,\dots,0$,
\begin{equation}
V_k(i)=\Phi(P_{k}(i,\cdot),V_{k+1}), \label{eq:discrete_KL_recursion_k}
\end{equation}
\end{subequations}
initialized by $V_{K}(i)=0$, for all $i\in X$.
\end{lemma}

\textit{Proof:} See Appendix~\ref{sec:path_KL_divergence}.

\begin{lemma}[Continuous KL Divergence to Uniform]
\label{lem:recursive_continuous_KL_to_uniform}
Given a Markov system $\mathcal{M}=(\mathcal{X},q_0,q,K)$, inducing trajectory distribution $T$~\eqref{eq:T}, the KL divergence to uniform admits the form $\operatorname{KL}(T\|U)=\phi(q_{0},\mathcal{V}_{0})$, where, for $k=K-1,K-2,\dots,0$, $\mathcal{V}_k(x)=\phi(q(x,\cdot),\mathcal{V}_{k+1})$, initialized by $\mathcal{V}_{K}(x)=0$, for all $x\in\mathcal{X}$.
\end{lemma}

\textit{Proof:} See Appendix~\ref{sec:path_KL_divergence_cont}.

\subsection{Local discretization difference bound}

\textit{Proof of Thm.~\ref{thm:local_varepsilon}:}
In this section, we combine Lemmas~\ref{lem:general_delta_upper_bound},~\ref{lem:recursive_dicrete_KL_to_uniform},
and~\ref{lem:recursive_continuous_KL_to_uniform} to bound the continuous-system
KL divergence to uniform through a corrected discrete-space dynamic programming
recursion.

For $x\in\mathcal{X}_i$, define the induced discrete transition probabilities $P_x(j) := \int_{\mathcal{X}_j} q(x,x^{\prime})\,\mathrm{d}x^{\prime}, \qquad j\in X$. By construction, $P_x(\cdot)\in\mathbb{P}_i$ for all $x\in\mathcal{X}_i$.

From Lemma~\ref{lem:general_delta_upper_bound}, we have, for all $x\in\mathcal{X}$, that
\[
\begin{split}
\phi(q(x,\cdot),0)-\Phi(P_x,0)&= \operatorname{KL}(q(x,\cdot)\|U_{\mathcal{X}})-\operatorname{KL}_D(P_x\|p^{u}_{\mathcal{X}}) \\
&\hspace{2cm}\le \varepsilon(P_x,n_{x},L_{\nabla q},\overline{\delta}).
\end{split}
\]
Rearranging gives
\[
\Phi(P_x,0)+\varepsilon(P_x,n_{x},L_{\nabla q},\overline{\delta})\ge \phi(q(x,\cdot),0), \qquad \forall x\in\mathcal{X}_i.
\]

\begin{lemma}
    If the following pointwise condition holds for all $i\in X$
    \begin{equation}
    \label{eq:pointwise_condition}
    \mathrm{(i)} \qquad 
    \Phi^{\varepsilon}(P_x,0) \;\ge\; \phi(q(x,\cdot),0), 
    \qquad \forall x\in\mathcal{X}_i,
    \end{equation}
    then~\eqref{eq:local_epsilon_recursion}, initialized by $V_{K}^{\varepsilon}(i)=0$, for all $i\in X$, ensures condition
    \[
        \mathrm{(ii)} \ \ \ V^{\varepsilon}_k(i)\ge \max_{x\in \mathcal{X}_{i}}\mathcal{V}_{k}(x),
    \]
    holds for all $i,k\in X\times\{0,1,\dots,K\}$, and (from this), we have
    \[
        \Phi^{\varepsilon}(\pi,V^{\varepsilon}_0)\ge \phi(q_0,\mathcal{V}_{0}).
    \]
\end{lemma}

We show that this implies condition~(ii). The claim is trivial for $k=K$.
Assume condition~(ii) holds at time $k+1$, and consider time $k$. Then, for all
$i\in X$,
\[
\begin{aligned}
&V_{k}^{\varepsilon}(i)= \max_{p^{\prime}\in\mathbb{P}_{i}}
\Phi^{\varepsilon}(p^{\prime},V_{k+1}^{\varepsilon}) \\
&= \max_{p^{\prime}\in\mathbb{P}_{i}}
\Big(\Phi^{\varepsilon}(p^{\prime},0)+\sum_{j\in X}p^{\prime}_{j}V_{k+1}^{\varepsilon}(j)\Big) \\
&\ge \max_{p^{\prime}\in\mathbb{P}_{i}}
\Big(\Phi^{\varepsilon}(p^{\prime},0)+\sum_{j\in X}p^{\prime}_{j}
\max_{x^{\prime}\in\mathcal{X}_j}\mathcal{V}_{k+1}(x^{\prime})\Big) \\
&\ge \max_{x\in\mathcal{X}_i}
\Big(\Phi^{\varepsilon}(P_x,0)+\sum_{j\in X}
\Big(\int_{\mathcal{X}_j} q(x,x^{\prime})\,\mathrm{d}x^{\prime}\Big)
\max_{x^{\prime\prime}\in\mathcal{X}_j}\mathcal{V}_{k+1}(x^{\prime\prime})\Big) \\
&\ge \max_{x\in\mathcal{X}_i}
\Big(\phi(q(x,\cdot),0)+\sum_{j\in X}
\int_{\mathcal{X}_j} q(x,x^{\prime})\,\max_{x^{\prime\prime}\in\mathcal{X}_j}\mathcal{V}_{k+1}(x^{\prime\prime})
\,\mathrm{d}x^{\prime}\Big) \\
&\ge \max_{x\in\mathcal{X}_i}
\Big(\phi(q(x,\cdot),0)+\sum_{j\in X}
\int_{\mathcal{X}_j} q(x,x^{\prime})\,\mathcal{V}_{k+1}(x^{\prime})\,\mathrm{d}x^{\prime}\Big) \\
&= \max_{x\in\mathcal{X}_i}
\Big(\phi(q(x,\cdot),0)+\int_{\mathcal{X}} q(x,x^{\prime})\,\mathcal{V}_{k+1}(x^{\prime})
\,\mathrm{d}x^{\prime}\Big)= \max_{x\in\mathcal{X}_i}\mathcal{V}_{k}(x).
\end{aligned}
\]
This establishes condition~(ii). The final claim
$\Phi^{\varepsilon}(\pi,V^{\varepsilon}_0)\ge \phi(q_0,\mathcal{V}_0)$ follows
analogously.
\qed


\section{Discrete KL divergence to uniform recursivity (Lemma~\ref{lem:recursive_dicrete_KL_to_uniform})}\label{sec:path_KL_divergence}

Let us define an alternative discrete KL divergence notation, denoting the divergence between two random variables $\mathsf{X}$ and $\mathsf{Y}$ on the same alphabet $X:=\{1,2,\dots,|S|\}$, as
\begin{equation}
\operatorname{KL}_{D}(\mathsf{X} \| \mathsf{Y}) :=\sum_{i\in X}\operatorname{Prob}(\mathsf{X} = i)\log \frac{\operatorname{Prob}(\mathsf{X}=i)}{\operatorname{Prob}(\mathsf{Y}=i)}.\label{eq:KL_divergence_definition_random_variables}
\end{equation}

\textit{Proof of Lemma \ref{lem:recursive_dicrete_KL_to_uniform}:} Let us start by proving, through induction, that 
\begin{equation}
V_0(i) = \operatorname{KL}_{D}(\mathsf{X}_{1:K}|\mathsf{X}_{0}=i \| \mathsf{U}_{1:K}|\mathsf{U}_0=i), \label{eq:recursion_proof_goal}
\end{equation}
for all $i\in X$, when iterating according to~\eqref{eq:discrete_KL_recursion_k}.

If for iteration $k+1$, we have 
\begin{equation}
V_{k+1}(i) = \operatorname{KL}_{D}(\mathsf{X}_{k+2:K}|\mathsf{X}_{k+1}=i\|\mathsf{U}_{k+2:K}|\mathsf{U}_{k+1}=i), \label{eq:recursion_induction_condition}
\end{equation}
for all $i\in X$, then for iteration $k$, using~\eqref{eq:Phi}, \eqref{eq:discrete_KL_recursion}, we have
\begin{equation}
\begin{split}
V_k(i) &= \underbrace{\sum_{j\in X}P_{k}(i,j)\log \frac{P_{k}(i,j)\lambda(\mathcal{X})}{\lambda(\mathcal{X}_j)}}_{\operatorname{KL}_{D}(\mathsf{X}_{k+1}|\mathsf{X}_{k}=i\|\mathsf{U}_{k+1}|\mathsf{U}_{k}=i)}\\
&\hspace{2cm}+\underbrace{\sum_{j\in X} P_{k}(i,j) V_{k+1}(j)}_{\operatorname{KL}_{D}(\mathsf{X}_{k+2:K}|\mathsf{X}_{k+1}\|\mathsf{U}_{k+2:K}|\mathsf{U}_{k+1})} \\
&=\operatorname{KL}_{D}(\mathsf{X}_{k+1:K}|\mathsf{X}_{k}=i\|\mathsf{U}_{k+1:K}|\mathsf{U}_{k}=i).
\end{split} \label{eq:recursion_induction}
\end{equation}

Secondly, as we initialize with $V_{K}(i)=0$, $ i\in X$, we have
\[
\begin{split}
V_{K-1}(i)&=\sum_{j\in X}P_{K-1}(i,j)\log \frac{P_{K-1}(i,j)\lambda(\mathcal{X})}{\lambda(\mathcal{X}_j)} \\
&= \operatorname{KL}_{D}(\mathsf{X}_{K}|\mathsf{X}_{K-1}=i\|\mathsf{U}_{K}|\mathsf{U}_{K-1}=i),
\end{split}
\]
satisfying~\eqref{eq:recursion_induction_condition} for $k=K-2$. Using identity~\eqref{eq:recursion_proof_goal}, we can now confirm identity~\eqref{eq:discrete_KL_recursion_0} by noting that
\[
\begin{split}
&\Phi(\pi,V_0)=\operatorname{KL}_{D}(\mathsf{X}_{0}\|\mathsf{U}_0)+ \operatorname{KL}_{D}(\mathsf{X}_{1:K}|\mathsf{X}_{0}\|\mathsf{U}_{1:K}|\mathsf{U}_0)\\
&\hspace{4cm}= \operatorname{KL}_{D}(\mathsf{X}_{0}\cdots\mathsf{X}_{K}\|\mathsf{U}_0\cdots \mathsf{U}_{K}),
\end{split}
\]
the last step following from the chain rule for conditional KL divergences.\qed

\section{Continuous KL divergence to uniform recursivity (Lemma~\ref{lem:recursive_continuous_KL_to_uniform})}\label{sec:path_KL_divergence_cont}

Let us define an alternative KL divergence notation, denoting the divergence between two random variables $\mathsf{X}$ and $\mathsf{Y}$ on the same compact alphabet $\mathcal{X}$, as
\begin{equation}
\operatorname{KL}(\mathsf{X} \| \mathsf{Y}) :=\int_{\mathcal{X}}q(\mathsf{X} = x)\log \frac{q(\mathsf{X}=x)}{q(\mathsf{Y}=x)}~\mathrm{d}x,\label{eq:KL_divergence_definition_random_variables_cont}
\end{equation}
where $\operatorname{Prob}(\mathsf{X}\in A)=\int_{A}q_{\mathsf{X}}(x)~\mathrm{d}x$, i.e., $q(\cdot)$ is the respective probability density. Notation $q(\mathsf{X}=x)$ and $q_{\mathsf{X}}(x)$ is used interchangeably to denote the same object.

\textit{Proof of Lemma~\ref{lem:recursive_continuous_KL_to_uniform}:} This proof is straightforwardly obtained by applying the exact steps from the Lemma~\ref{lem:recursive_dicrete_KL_to_uniform} proof to the continuous conditions described above~\eqref{eq:KL_divergence_definition_random_variables_cont}. \qed

\section{Bounding the gradient (Lemma~\ref{lem:bounding_the_gradient})} \label{apx:bounding_the_gradient_proof}

\textit{Proof of Lemma~\ref{lem:bounding_the_gradient}:} Note that $\mathcal{S}=\mathcal{X}^{K+1}$, and thus we can interpret trajectory $s\in \mathcal{S}$ as $s=x_{0}x_{1}\cdots x_{K}$. The derivative of $T(x_0x_1\cdots x_{K}) = q_0(x_0)\prod_{k=1}^{K}q(x_{k-1},x_k)$, w.r.t.\ element $j\in\{1,2,\dots,n_{x}\}$ of  $x_{k,j}$, for $k\in\{1,2,\dots,K-1\}$, is found as
\[
\begin{split}
&\nabla_{x_{k,j}}T = q_0(x_0) \Big(\hspace{2mm} \prod_{\mathclap{t\in\{1,\dots,K\}\setminus\{k\}}} \ q(x_{t-1},x_t)\Big)\frac{\partial }{\partial x_{k,j}}q(x_{k-1},x_{k})\\
&\hspace{1.9cm}+ q_0(x_0)\Big(\hspace{2mm} \prod_{\mathclap{t\in\{1,\dots,K\}\setminus\{k+1\}}} \ q(x_{t-1},x_t)\Big)\frac{\partial }{\partial x_{k,j}}q(x_{k},x_{k+1}).
\end{split}
\]
Given the bounds $L_{\nabla q}\in\mathbb{R}_{\ge 0 }$ and $L_{q}\in\mathbb{R}_{\ge 0 }$~\eqref{eq:bounds_on_markov_kernels}, for all $j\in\{1,2,\dots,n_{x}\}$, we upper-bound the above equation as
\[
\|\nabla_{x_{k,j}}T(x_{0:K})\|_{\infty} \le 2(L_{q})^{K}L_{\nabla q},
\]
for any $k\in\{1,2,\dots,K-1\}$. Repeating the above steps for $k\in\{0,K\}$ yields $\nabla_{x_{k},j}T(x_{0:K}) \le (L_{q})^{K}L_{\nabla q}$, for $k\in\{0,K\}$. \qed

\end{document}